\documentclass[conference, 9pt]{IEEEtran}
\IEEEoverridecommandlockouts

\usepackage{calc,amsfonts,amssymb,amsmath,bm,url,color,graphicx,cite,shortcuts_OPT,enumitem,tcolorbox,booktabs}
\usepackage{psfrag,float,hyperref,bbm}
\hypersetup{
  colorlinks   = true, 
  urlcolor     = blue, 
  linkcolor    = blue, 
  citecolor    = blue   
}
\usepackage{algorithm}
\usepackage{algorithmic}
\usepackage{amsthm}

\newtheorem{Lemma}{Lemma}
\newtheorem{Prop}{Proposition}
\newtheorem{Theorem}{Theorem}
\newtheorem*{Theorem*}{Theorem}
\newtheorem{Def}{Definition}

\newtheorem{Exa}{Example}
\newtheorem{assumption}{H\!\!}

\usepackage{pgfplots}
\usetikzlibrary{arrows,shapes,calc,tikzmark,backgrounds,matrix,decorations.markings}
\usepgfplotslibrary{groupplots}
\usepgfplotslibrary{fillbetween}

\pgfplotsset{compat=1.3}
\usepackage{subcaption}
\usepackage{graphicx} 

\definecolor{lavander}{cmyk}{0,0.48,0,0}
\definecolor{violet}{cmyk}{0.79,0.88,0,0}
\definecolor{burntorange}{cmyk}{0,0.52,1,0}

\definecolor{asuorange}{rgb}{1,0.699,0.0625}
\definecolor{asured}{rgb}{0.598,0,0.199}
\definecolor{asuborder}{rgb}{0.953,0.484,0}
\definecolor{asugrey}{rgb}{0.309,0.332,0.340}
\definecolor{asublue}{rgb}{0,0.555,0.836}
\definecolor{asugold}{rgb}{1,0.777,0.008}

\title{Graph Learning with Network Games Prior}

    \makeatletter
    \def\multilimits@{\bgroup
  \Let@
  \restore@math@cr
  \default@tag
 \baselineskip\fontdimen10 \scriptfont\tw@
 \advance\baselineskip\fontdimen12 \scriptfont\tw@
 \lineskip\thr@@\fontdimen8 \scriptfont\thr@@
 \lineskiplimit\lineskip
 \vbox\bgroup\ialign\bgroup\hfil$\m@th\scriptstyle{##}$\hfil\crcr}
    \def\Sb{_\multilimits@}
    \def\endSb{\crcr\egroup\egroup\egroup}
\makeatother

\makeatletter
\DeclareRobustCommand*\cal{\@fontswitch\relax\mathcal}
\makeatother

\title{Network Games Induced Prior for Graph Topology Learning}

\author{ 
\IEEEauthorblockN{Chenyue Zhang, Shangyuan Liu, Hoi-To Wai, Anthony Man-Cho So}
\IEEEauthorblockA{\textit{Dept. of Systems Engineering \& Engineering Management}, \textit{The Chinese University of Hong Kong}\\
\texttt{\{czhang, htwai, manchoso\}@se.cuhk.edu.hk}, \texttt{shangyuanliu@link.cuhk.edu.hk}}
}

\date{}

\begin{document}
\maketitle

\begin{abstract} 
Learning the graph topology of a complex network is challenging due to limited data availability and imprecise data models. A common remedy in existing works is to incorporate priors such as sparsity or modularity which highlight on the structural property of graph topology. We depart from these approaches to develop priors that are directly inspired by complex network dynamics. Focusing on social networks with actions modeled by equilibriums of linear quadratic games, we postulate that the social network topologies are optimized with respect to a social welfare function. 
Utilizing this prior knowledge, we propose a network games induced regularizer to assist graph learning. 
We then formulate the graph topology learning problem as a bilevel program. We develop a two-timescale gradient algorithm to tackle the latter. We draw theoretical insights on the optimal graph structure of the bilevel program and show that they agree with the topology in several man-made networks. Empirically, we demonstrate the proposed formulation gives rise to reliable estimate of graph topology.
\end{abstract}


\section{Introduction}
Graph-based structures are increasingly utilized in data science to represent relationships among features and datasets. In particular, graph representations are pivotal for unveiling relational networks and supporting diverse learning tasks like graph neural networks, sampling, semi-supervised learning, and graph signal processing (GSP). In many applications, graph topologies are neither immediately available nor easily discernible. This necessitates inferring graph topologies from node observations, a.k.a.~graph signals. Such problem, also known as {graph (topology) learning}, has garnered attention in machine learning and signal processing. Recent works have developed graph learning algorithms utilizing GSP models, leveraging properties such as smoothness and stationarity \cite{dong2019learning, mateos2019connecting,liu2023logspect}.


Learning the graph topology of complex networks is challenging due to limited data availability and imprecise data models. 
To tackle these issues, a common practice is to incorporate \emph{prior information} to assist graph learning via regularization. 
Notably, as demonstrated in \cite{friedman2008sparse, dong2016learning, kalofolias2016learn, egilmez2017graph, kumar2020unified, nie2016constrained}, integrating carefully crafted regularizers can induce desired graph structures—like sparse, modular, bipartite, regular graphs—within the graph learning paradigm. However, applying such designs requires a-priori knowledge on the \emph{graph topology} which are often acquired in a heuristic fashion. For instance, graphs of certain networks are found to be sparse or modular as observed from the patterns in a number of real world networks. 

An alternative perspective is that graph topologies may be viewed as the optimized solutions with respect to a function/task of the latent network dynamics. Considering the case of social networks, actions of individuals can be modeled as the equilibrium of a network game determined by the graph topology \cite{jackson2015games}. The `optimality' of a graph topology can be measured thru the equilibrium actions---a natural choice is the \emph{total social welfare} function taken as the aggregated equilibrium action \cite{demange2017optimal}. 
As a motivating example, Table~\ref{table:ReaNetWkRewir} shows the average performance loss in total social welfare of perturbing several graph topologies under \emph{random rewiring} (details in Example~\ref{ex:rewire}). 
We find that man-made or social networks ({\tt WikiVote}, {\tt Karate}) are more sensitive to random rewiring than the non-man-made ones ({\tt Dolphins}). Our observation suggests that man-made or social networks may be optimized for maximum total social welfare.

\begin{table}[t]
\centering
\includegraphics[width=0.47\textwidth]{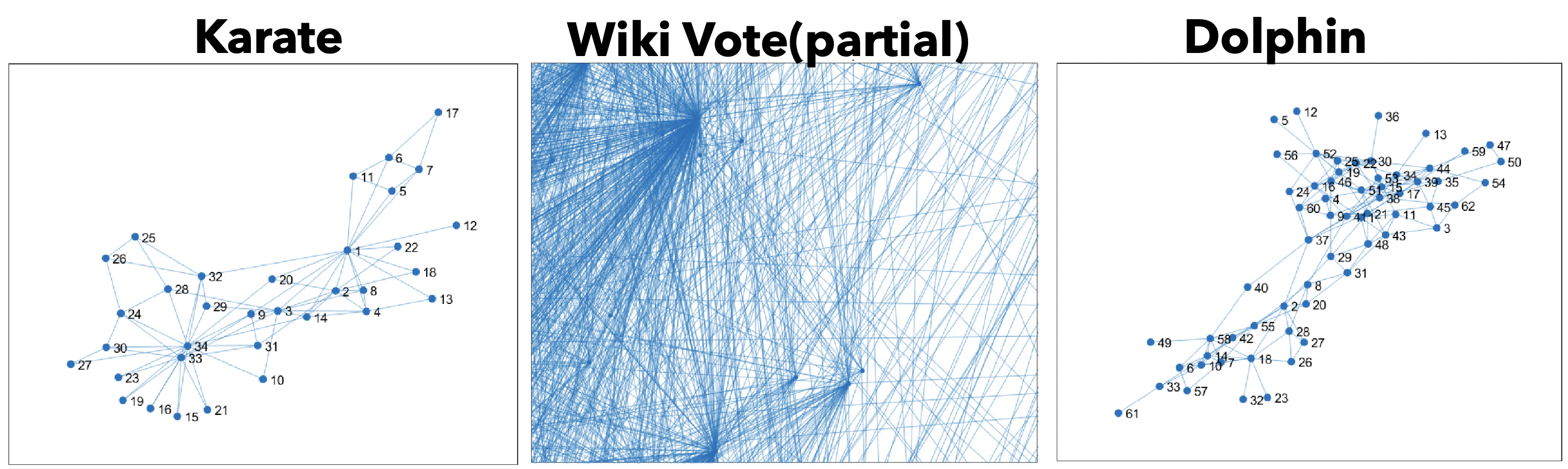}\vspace{.2cm}
\begin{tabular}{l|l|l|l|l|l} 
\toprule 
Rewiring & $10\%$ & $20\%$ & $30\%$ & $40\%$ & $50\%$ \\
\midrule
{\tt Karate} & $\mathbf{94.06\%}$ & $\mathbf{88.86\%}$ & $\mathbf{84.72\%}$ & $\mathbf{81.17\%}$ & $\mathbf{78.30\%}$ \\
{\tt WikiVote} & $\mathbf{96.28\%}$ & $\mathbf{93.07\%}$ & $\mathbf{90.32\%}$ & $\mathbf{88.01\%}$ & $\mathbf{86.11\%}$ \\
{\tt Dolphins} & $98.15\%$ & $96.48\%$ & $95.13\%$ & $93.97\%$ & $93.08\%$ \\
\bottomrule
\end{tabular}
\caption{(Top) Graphs of tested networks. (Bottom) Impact of random rewiring on  total social welfare. See Example~\ref{ex:rewire} of Sec.~\ref{sec:main}.
}\label{table:ReaNetWkRewir}\vspace{-.2cm}
\end{table}

This paper leverages the total social welfare as a {\it functional prior} for graph learning. It departs from the conventional structural prior approach with an emphasis on knowledge of the {\it tasks performed by networks} and suggests to regulate graph learning by adjusting the latent network dynamics that drives the prior.
As the first step towards treating general functional priors, our contributions are:
\begin{itemize}[leftmargin=*]
\item We propose and formulate a graph learning problem from smooth signals with a network game induced prior (GLGP). The latter is formulated as a bilevel program, which is tackled by developing a two-timescales gradient (TTGD) algorithm inspired by \cite{hong2023two}.
\item We analyze the GLGP problem and provide theoretical insights on the graphs learnt from maximizing the social welfare. Particularly, we show that it induces graph topology with a few hub nodes, which coincides with the examples in Table~\ref{table:ReaNetWkRewir}. 
\item We show that the proposed TTGD algorithm finds an ${\cal O}(1/K)$-stationary solution to the GLGP problem after $K$ iterations. 
\end{itemize}
Finally, we present numerical experiments on synthetic and real data. The proposed GLGP problem learns topologies that are optimized in a task-based manner while respecting the (smooth) graph signal observations. In settings with insufficient data, the proposed method outperforms several state-of-the-art graph learning algorithms.

\vspace{.1cm}
\noindent {\bf Related Works.} 
In addition to the above cited works, network games have been used for modeling \emph{graph data} and have inspired the graph learning problems in \cite{shen2017kernel, wai2019joint, leng2020learning, rossi2022learning}, as well as low pass graph signal models \cite{ramakrishna2020user, zhang2024detecting}. Our work differs from them as we treat the network games only as the latent dynamics driving the functional prior, while the graph learning component depends on the general property of smoothness.
Lastly, we notice that recent works have considered non-structural priors, e.g., \cite{pu2021learning, wasserman2024graph} applied learning-to-optimize to learn a prior function from graph templates, \cite{niresi2024informed} studied the use of domain knowledge, \cite{navarro2024fair} analyzed fairness as a prior.

\section{Preliminary: Graph Learning from Smooth Signals} 
We review the graph learning framework from smooth signals in \cite{kalofolias2016learn}. To fix notations, we consider a networked system characterized by the possibly directed graph ${\cal G} = (V,E)$ with the node set $V = [N] := \{1,\ldots,N\}$ and edge set $E \subseteq V \times V$. The ordered tuple $(i,j) \in E$ indicates an edge from $i$ to $j$. The graph ${\cal G}$ is endowed with a (possibly non-symmetric) {adjacency matrix}, $\GSO \in \RR^{N \times N}$ such that $W_{ij} \neq 0$ if and only if $(j,i) \in E$. 

Our goal is to learn the matrix $\GSO$ given a dataset of $M$ {graph signals} denoted by ${\bm X} := ( {\bm x}_1, \ldots, {\bm x}_M ) \in \RR^{N \times M}$. The graph signals are assumed to be \emph{smooth} such that the signal values, $x_{m,i}, x_{m,j}$, are similar between pairs of adjacent nodes, i.e., if $(i,j) \in E$ or $(j,i) \in E$. 
For example, this corresponds to the case where the graph signals are the output of a low pass graph filter \cite{ramakrishna2020user}, e.g., heat diffusion process, DeGroot opinion dynamics.

Consider the following convex optimization problem for learning the graph topology $\GSO$ inspired by \cite{kalofolias2016learn}: 
\begin{align}
\displaystyle \min_{ \GSO  } & \displaystyle ~J(\GSO; {\bm X}) := \frac{1}{2M} \sum_{i,j=1}^N W_{ij} \| {\bm x}_i^{\rm row} - {\bm x}_j^{\rm row} \|^2 + \beta \| \GSO \|_{\rm F}^2 \label{eq:graph_learn} \\
{\rm s.t.} & ~\GSO \in {\cal S} := \{ \GSO : \GSO \geq {\bm 0}, {\rm diag}( \GSO ) = {\bm 0}, \GSO {\bm 1}=c{\bm 1}\}, \notag
\end{align}
where ${\bm x}_i^{\rm row} \in \RR^M$ denotes the $i$th row vector of ${\bm X}$, $c, \beta > 0$ are regularization parameters and ${\cal S}$ is the set of admissible adjacency matrices. Similar formulations to \eqref{eq:graph_learn} include the GLSigRep model \cite{dong2016learning} which can be considered as a special case of \eqref{eq:graph_learn} as discussed in \cite{kalofolias2016learn}. 
This allows one to connect \eqref{eq:graph_learn} to the MAP estimation of $\GSO$ under a graph filter model \cite{dong2016learning}. In particular, with sufficient number of observations ($M \gg 1$), solving the graph learning problem \eqref{eq:graph_learn} yields an accurate estimate of the ground truth topology. 
We remark that in \cite{kalofolias2016learn}, the above formulation was developed for the case when $\GSO$ is symmetric. Moreover, alternative formulations can be applied, e.g., via exploiting stationary graph signals \cite{segarra2017network}, structural equation model \cite{shen2017kernel}; see the overviews in \cite{mateos2019connecting, dong2019learning}. 

In practice, high quality smooth graph signals are rare for large networked systems. Typically, $M \ll N$. It is customary to introduce an additional regularizer to \eqref{eq:graph_learn}. Let $\lambda > 0$, we consider
\begin{equation} \label{eq:graph_learn_prob}
\textstyle \min_{ \GSO \in {\cal S} }~J( \GSO; {\bm X} ) + \lambda \, {\cal R}( \GSO ).
\end{equation}
Common designs for the regularizer ${\cal R}( \GSO )$ impose \emph{prior knowledge} about the \emph{structure of graph topology}. Examples include inducing a sparse graph via setting ${\cal R}(\GSO) = \| \GSO \|_1$ \cite{kalofolias2016learn}, inducing a modular graph with $K$ densely connected components via setting {${\cal R}( \GSO ) = \sum_{i=K+1}^N \sigma_i( \GSO )$}, where $\sigma_i(\cdot)$ denotes the $i$th largest eigenvalues of $\GSO$ \cite{nie2016constrained}; also see \cite{kumar2020unified}. Although these  are intuitive designs, they may fail to capture the intricate features in the graph topology.

As mentioned in the Introduction, this work focuses on a novel class of regularizer for graph topology learning using \emph{prior information} modeled directly by the \emph{networked system dynamics}. The next section will introduce a {functional regularizer} induced by \emph{network games} inspired by the economics literature. 

\section{Graph Learning with Network Games Prior} \label{sec:main}
We will first introduce a general linear-quadratic network game and discuss its basic properties. Then, we will develop a network games-inspired regularizer ${\cal R}( \GSO )$ as an implicit function of $\GSO$ and justify its effectiveness using real network examples. Finally, a two-timescale gradient descent (TTGD) algorithm will be derived to tackle the resultant bilevel optimization problem.

\noindent \textbf{Network Games.} Let $i \in V$ denotes the $i$th individual/agent in a social network described by $\GSO$, $y_i \in \RR_+$ be the action selected by the $i$th agent, and ${\bm y}_{-i}$ denotes the vector ${\bm y} = (y_1, \ldots, y_N)$ with the $i$-th element removed. We concentrate on the network games with linear-quadratic payoffs and strategic complements \cite{jackson2015games}. For any $i \in V$, the $i$th agent selects an action $y_i$ to maximize a payoff function depending on the neighbors' actions:
\beq \textstyle \label{eq:lq_game}
U_i( y_i ; {\bm y}_{-i} ; \GSO ) = - \frac{y_i^2}{2} + y_i \left( \sum_{j=1, j \neq i}^N W_{ij} f(y_j) + b_i \right),
\eeq 
where $b_i \geq 0$ is the marginal benefit and $f : \RR_+ \to \RR_+$ is an \emph{interaction function} such that an increase in neighbor's action can positively affect $y_i$ \cite{cai2024optimal}. Throughout, we assume $\GSO \in {\cal S}$ and
\begin{assumption} \label{assu:lipsf}
The function $f: \RR_+ \to \RR_+$ is $1$-Lipschitz, twice differentiable, non-decreasing, and concave. Moreover, it holds that $f(0)=0$, $f(x) \leq x$, $|f''(x)| \leq 1$ for any $x \geq 0$, and $c < 1$.
\end{assumption}

Given $\GSO$, we are interested in the \emph{Nash Equilibrium} (NE) strategy of \eqref{eq:lq_game}, ${\bm y}^{\sf NE}( \GSO )$, which is a set of actions taken by agents in $V$ where no agent  shall change his/her action. In other words, it holds that 
\beq \textstyle \label{eq:ne_form}
\begin{aligned}
y_i^{\sf NE} = {\sf T}_i( {\bm y}^{\sf NE}; \GSO ) & \textstyle := \argmax_{ y_i \geq 0} U_i( y_i ; {\bm y}_{-i}^{\sf NE} ; \GSO ) \\
& \textstyle = \max\{0 , b_i + \sum_{j=1, j \neq i}^N W_{ij} f( y_j^{\sf NE} ) \} ,
\end{aligned}
\eeq 
for all $i \in V$.
It is proven in \cite{cai2024optimal} that under H\ref{assu:lipsf}, $\GSO \in {\cal S}$, the NE is well-defined and is the fixed point of ${\sf T} ( {\bm y}; \GSO ) := ( {\sf T}_1( {\bm y}; \GSO ) ; \cdots ; {\sf T}_N ( {\bm y}; \GSO ) )$. 
For example, the actions in an NE strategy corresponds to the intensity of economics activities \cite{Candogan}. 

The NE strategy induces a performance metric of the candidate network $\GSO$. A reasonable metric is the \emph{total social welfare} \cite{demange2017optimal}:\vspace{-.2cm}
\beq 
{\sf Wel}( \GSO ) := {\bf 1}^\top {\bm y}^{\sf NE}( \GSO ) \vspace{-.2cm}
\eeq 
such that ${\bm y}^{\sf NE}(\cdot)$ satisfies \eqref{eq:ne_form}. We conjecture that
\begin{center}
\emph{{\bf (C)}~the topologies of man-made networks (e.g., social networks) are self-optimized for maximum ${\sf Wel}( \GSO )$.}
\end{center}
In general, we conjecture that network topologies are \emph{task optimized}.
Although it remains an open problem to verify the conjecture, empirically we found that several man-made networks exhibit traits of being self-optimized w.r.t.~${\sf Wel}( \GSO )$.
\begin{Exa} \label{ex:rewire}
We revisit the motivating example in Table~\ref{table:ReaNetWkRewir}. In detail, this example evaluates ${\sf Wel}( \GSO )$ subject to various levels of random rewiring. Fixing {${\bm b} = {\bm 1}$}, the table shows the welfare ratio
\beq 
P_{\sf pert.} = \mathbb{E} \left[ \frac{ {\sf Wel}( \GSO_{\sf pert.} ) - {\bf 1}^\top {\bm b} }{ {\sf Wel}( \GSO_{\sf orig.} ) - {\bf 1}^\top {\bm b} } \right],
\eeq 
where $\GSO_{\sf orig.}$, $\GSO_{\sf pert.}$ are the {original, perturbed binary adjacency matrices}, respectively.
Under conjecture {\bf (C)}, we anticipate that self-optimized networks will suffer from greater drop in $P_{\sf pert.}$ than non-optimized ones with the same proportion of random rewiring. The table supports the conjecture by comparing the performances of the \texttt{WikiVote}, \texttt{Karate} networks (man-made) against the {\tt Dolphins} network (non-man-made).
\end{Exa}
\noindent We refer the readers to \cite{demange2017optimal, sanhedrai2022reviving, meena2023emergent} for related observations on how real world network topologies show traits of self-optimization.  

\vspace{.1cm}
\noindent \textbf{Graph Learning with Network Games Prior.} Suppose that the ground truth graph ${\cal G}$ is a man-made network. Our next step is to incorporate the social welfare prior discussed above into the graph learning formulation from smooth signals. 

Under conjecture {\bf (C)}, it is natural to take the regularizer as the total social welfare function, i.e., ${\cal R}( \GSO ) = - {\sf Wel}( \GSO )$. Substituting the above into \eqref{eq:graph_learn_prob} gives rise to the following graph learning with network games prior \eqref{eq:glfp} problem:
\begin{align}
\displaystyle \min_{ \GSO, {\bm y} } & ~~\Phi( \GSO, {\bm y} ) := J( \GSO; {\bm X} ) - \lambda {\bf 1}^\top {\bm y} \label{eq:glfp} \tag{GLGP} \\
\text{s.t.} & ~~\textstyle y_i \in \argmax_{ \hat{y}_i \in \RR_+ } U_i( \hat{y}_i ; {\bm y}_{-i} ; \GSO ),~\forall~i \in V,~\GSO \in {\cal S}. \notag
\end{align} 
The proposed GLGP formulation utilizes knowledge from both smooth graph signals and the social welfare maximizing property of man-made networks. As ${\sf Wel}( \GSO )$ depends implicitly on $\GSO$, \eqref{eq:glfp} is a \emph{bilevel program} with a variational inequality constraint. 

\subsection{Structural Interpretation for the Network Games Prior}
Note that \eqref{eq:glfp} takes the total social welfare ${\sf Wel}( \GSO )$ to measure the network's performance and inform graph learning from a \emph{task oriented} perspective. To further investigate properties of the network games prior, this subsection offers a theory-guided interpretation. 

Our idea is to consider relaxation of the bilevel optimization problem and analyze the KKT condition of the relaxed problem. First, we consider the simplified problem:\vspace{-.1cm}
\begin{equation} \label{eq:net-bilevel-app} \textstyle
\min_{ \GSO \in {\cal S} } J( \GSO; {\bm X} ) - \lambda {\bf 1}^\top \GSO {\bm b}. \vspace{-.1cm}
\end{equation}
\begin{Prop} \label{prop:approx}
Under H\ref{assu:lipsf}, problem \eqref{eq:net-bilevel-app} is equivalent to a relaxed version of problem \eqref{eq:glfp}.
\end{Prop}
\noindent The proposition is obtained by rewriting the NE condition as nonlinear equality constraints and relaxing the latter using H\ref{assu:lipsf}. 

The above proposition shows that \eqref{eq:net-bilevel-app} serves as a simplified surrogate for analyzing the optimal solution of \eqref{eq:glfp}. Particularly, 
\begin{Prop} \label{prop:kkt}
There exists $\eta_i \in \RR, i \in V$ such that any optimal solution to Problem \eqref{eq:net-bilevel-app} is given by \vspace{-.2cm}
\beq \label{eq:netbilevel-sol}
    W_{ij}^\star = {\textstyle \frac{1}{2 \beta}} \max \Big\{ 0, \lambda b_j + \eta_i -\frac{ \| {\bm x}_i^{\rm row} - {\bm x}_j^{\rm row} \|^2}{2M} \Big\},  
\eeq
for any $i \neq j$ and $W_{ii}^\star = 0$.
It also holds that $\GSO^\star {\bf 1} = c {\bf 1}$.
\end{Prop}
\noindent The proof is obtained by analyzing the KKT conditions of \eqref{eq:net-bilevel-app}.

To study \eqref{eq:netbilevel-sol}, we consider the case when $\lambda \gg 1$. The proposition above shows that $W_{ij}^\star \approx \frac{\lambda}{2\beta} b_j$ for any $i \in V$. In this way, any optimal solution to \eqref{eq:net-bilevel-app} gives a graph topology that exhibits a \emph{`hub' structure} where edges are emanated from nodes with large $b_j$. In fact, this observation coincides with the man-made network examples studied in Table~\ref{table:ReaNetWkRewir} where we observe a number of `hub' nodes.





\subsection{Efficient Algorithm for Tackling \eqref{eq:glfp}}
The next endeavor is to derive an efficient algorithm for tackling \eqref{eq:glfp}.
Here, the problem entails a lower level subproblem that requires solving the NE given the candidate $\GSO$, and an upper level subproblem depending on the computed NE ${\bm y}$. While the lower level problem can be solved by fixed point iteration \cite{cai2024optimal} (under H\ref{assu:lipsf}), the overall bilevel problem remains non-convex in general.

Our idea is to develop a two-timescale gradient (TTGD) algorithm similar to \cite{hong2023two}. We first notice that under H\ref{assu:lipsf}, 
\eqref{eq:glfp} is equivalent to a single level optimization problem to minimize $\ell( \GSO ) := \Phi( \GSO; {\bm y}^{\sf NE} ( \GSO ) )$. A common solution is to apply the projected gradient descent algorithm: let $\gamma>0$ be the step size,
\beq \label{eq:proj_grad}
\GSO^{k+1} = {\sf Proj}_{ \cal S } ( \GSO^k - \gamma \grd \ell( \GSO^k ) ),~\forall~k \geq 0,
\eeq 
where ${\sf Proj}_{\cal S}(\cdot)$ denotes the Euclidean projection onto ${\cal S}$.
The challenge of \eqref{eq:proj_grad} lies in the gradient computation $\grd \ell( \GSO^k )$ since $\ell( \GSO )$ has an implicit dependence on $\GSO$ via the NE map ${\bm y}^{\sf NE}(\cdot)$. Let $\bar{\bm y}^k := {\bm y}^{\sf NE}( \GSO^k )$, it can be shown that \cite{liu2022inducing}
\begin{align}
& \grd \ell( \GSO^k ) = \widehat{\grd} \Phi( \GSO^k , \bar{\bm y}^k ) := \grd_{\GSO} \Phi( \GSO^k; \bar{\bm y}^k ) \label{eq:grad_ideal} \\
& - ( {\rm J}_{\GSO} {\sf T}( \bar{\bm y}^k; \GSO^k) )^\top ( {\rm J}_{\bm y} {\sf T}( \bar{\bm y}^k; \GSO^k) - {\bm I}_N )^{-\top} \grd_{\bm y} \Phi( \GSO^k; \bar{\bm y}^k ), \notag
\end{align} 
where ${\rm J}_{\GSO} {\sf T}(\cdot), {\rm J}_{\bm y} {\sf T}(\cdot)$ denotes the Jacobian of the operator ${\sf T}: \RR^N \to \RR^N$ in \eqref{eq:ne_form} w.r.t.~$\GSO, {\bm y}$, respectively, and $\grd_{\GSO} \Phi(\cdot), \grd_{\bm y} \Phi(\cdot)$ denotes the partial gradient taken w.r.t.~$\GSO, {\bm y}$, respectively.

However, evaluating $\grd \ell( \GSO^k )$ requires the NE strategy ${\bm y}^{\sf NE}( \GSO^k )$, where the latter may not be available in closed form. That said, when $\GSO^k$ is fixed, the fixed point iteration ${\bm y} \leftarrow {\sf T}( {\bm y}; \GSO^k )$ finds the NE strategy at a linear rate \cite{parise2019variational}. As inspired by \cite{hong2023two} and let $\gamma, \alpha > 0$ be the stepsizes, we consider:
\begin{tcolorbox}[boxsep=2pt,left=4pt,right=4pt,top=3pt,bottom=3pt, 
]
\underline{\emph{Two Timescale Gradient (TTGD) procedure}}: for $k \geq 0$,
\vspace{-.0cm}\begin{subequations} \label{eq:ttgd}
\begin{align}
{\bm y}^{k+1} & = {\bm y}^k + \alpha ( {\sf T}( {\bm y}^k ; \GSO^k ) - {\bm y}^k ), \label{eq:ttgd_y} \\
\GSO^{k+1} & = {\sf Proj}_{ \cal S } ( \GSO^k - \gamma \widehat{\grd} \Phi( \GSO^k , {\bm y}^{k+1} ) ). \vspace{-.1cm} \label{eq:ttgd_W}
\end{align}
\end{subequations} 
\end{tcolorbox}
Comparing to \eqref{eq:proj_grad}, the TTGD procedure uses an inexact version of the gradient of $\ell$ in \eqref{eq:ttgd_W} evaluated on ${\bm y}^k$. Here, the intuition is that when $\gamma \ll \alpha$,  $\GSO^k$ will appear to be `static' w.r.t.~the update of NE strategy \eqref{eq:ttgd_y}. By adjusting $\gamma, \alpha$, we can ensure that the latter finds a close approximation for ${\bm y}^{\sf NE}( \GSO^k )$ and therefore the algorithm converges. Formally, we observe that: 
\begin{Theorem} \label{thm:ttgd}
Assume H\ref{assu:lipsf}, then there exists a set of step sizes with {$\alpha = \frac{1-c}{(1+c)^2}$}, $\gamma \ll \alpha$, such that 
it holds for any $K \geq 1$,
\beq \notag
\min_{k=1,...,K} \| \gamma^{-1} ( \GSO^k - {\sf Proj}_{\cal S} ( \GSO^k - \gamma \grd \ell( \GSO^k ) ) ) \|^2 = {\cal O}(K^{-1}).
\eeq 
\end{Theorem}


\noindent 
Notice that H\ref{assu:lipsf} guarantees that the fixed point map satisfies $\| {\sf T}( {\bm y}; \GSO ) - {\sf T}( {\bm y}' ; \GSO ) \| \leq c \| {\bm y} - {\bm y}' \|$ with $c < 1$ for any ${\bm y}, {\bm y}' \in \RR^N$. Consequently, the Lipschitz-ness of $\grd \ell( \GSO )$ can be established using H\ref{assu:lipsf} and other properties of the NE map. The rest of our analysis follows from adapting the framework in \cite{hong2023two} to the case where the lower level subproblem involves variational inequalities that are solved in a deterministic fashion. 

Lastly, we comment on the computation complexity of \eqref{eq:ttgd}. The key bottleneck lies in computing the hyper-gradient in \eqref{eq:grad_ideal} involving $({\rm J}_{\bm y} {\sf T}( {\bm y}^k; \GSO^k) - {\bm I}_N)^{-1}$. This entails a complexity of ${\cal O}( N^3 )$ FLOPs per iteration. Nevertheless, we envision that the algorithm can be accelerated using penalty based algorithms such as \cite{ji2021bilevel, kwon2023fully}. 

\section{Numerical Experiments}
The last section presents numerical experiments on synthetic and real data to validate the efficacy of \eqref{eq:glfp}.

\subsection{Synthetic Data}
Our first experiment aims at evaluating the graph learning performance for graph signals generated from a preferential attachment (PA) graph {with one edge to attach for every new node}, 
and \( N = 50 \) nodes \cite{newman2018networks}. 
The probability of a new edge linking to an existing node is proportional to its degree relative to the total degree of all nodes.
We concentrate on a scenario with limited data acquired where only $M = 10 \ll N$ smooth graph signals are observed. Each graph signal is generated via a low pass graph filter, ${\cal H}(\GSO) = \exp({ \GSO /2} )$, as \( \bm{x}_m = \exp({ \GSO /2} ) \bm{u}_m+\bm{w}_m \), where \( {\bm u}_m \sim {\cal N}( 0, {\bm I} ) \) is an i.i.d.~white noise excitation and the observation/modeling noise follows ${\bm w}_m \sim N\left(\mathbf{0}, \sigma^{2} \I\right) $ with $\sigma=0.2$. 
We also fix \( (\beta,c) = (200,0.95) \) for \eqref{eq:graph_learn}.
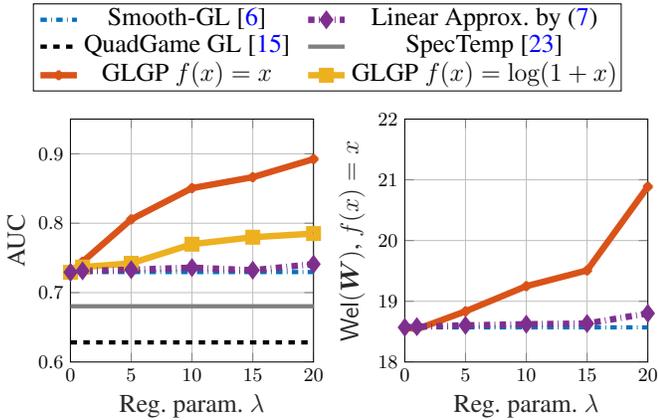
\begin{figure}[t]
\begin{center}
\hspace{-1.2cm}\resizebox{0.95\linewidth}{!}{\definecolor{mycolor1}{rgb}{0.00000,0.44700,0.74100}%
\definecolor{mycolor2}{rgb}{0.85000,0.32500,0.09800}%
\definecolor{mycolor3}{rgb}{0.92900,0.69400,0.12500}%
\definecolor{mycolor4}{rgb}{0.49400,0.18400,0.55600}%

\begin{tikzpicture}
\begin{groupplot}[group style={group name=myplot,group size=2 by 1}]
\nextgroupplot[%
width=4cm,
height=4cm,
scale only axis,
xmin=0,
xmax=20,
xlabel style={font=\color{white!15!black}},
xlabel={\Large Reg.~param.~$\lambda$},
ymin=0.6,
ymax=.95,
ylabel style={font=\color{white!15!black}},
ylabel={\Large AUC},
xtick={0,5,10,15,20}, 
axis background/.style={fill=white},
xmajorgrids,
ymajorgrids,
legend pos = north west,
legend style={
legend cell align=left, align=left, 
  draw opacity=0.8,
  text opacity=1,draw=white!15!black,font=\large},
]

\addplot [color=mycolor1, dashdotted, line width=2.0pt]
  table[row sep=crcr]{%
0.01	0.729628477119408\\
20	0.729628477119408\\
};\label{plots:Smooth-GL}

\addplot [color=black, dashed, line width=2.0pt]
  table[row sep=crcr]{%
0.01	0.628264711379973\\
20	0.628264711379973\\
};\label{plots:xiaowen}

\addplot [color=gray, line width=2.0pt]
  table[row sep=crcr]{%
0.01	0.680347499532703\\
20	0.680347499532703\\
};\label{plots:spec}

\addplot [color=mycolor2, line width=3.0pt, mark=x, mark options={solid, mycolor2}]
  table[row sep=crcr]{%
0.01	0.729273976618125\\
1	0.744401561623817\\
5	0.805497863175246\\
10	0.850517425954564\\
15  0.866268543220787\\
20	0.892368668116708\\
};\label{plots:bi-x}

\addplot [color=mycolor3, line width=3.0pt, mark=square, mark options={solid, mycolor3}]
  table[row sep=crcr]{%
0.01	0.72919127955133\\
1	0.736776219646894\\
5	0.742219387755103\\
10	0.769763398698365\\\\
15  0.779794240343932\\
20	0.785076084555388\\
};\label{plots:bi-log}

\addplot [color=mycolor4, dashdotted, line width=3.0pt, mark=diamond, mark options={solid, mycolor4}]
  table[row sep=crcr]{%
0.01	0.72929223355231\\
1	0.731332762211272\\
5	0.7329382667773859\\
10	0.7359933241926265\\
15  0.7322342801751782\\
20	0.741021146100323\\
};\label{plots:bi-app}
\hspace{.5cm}

\nextgroupplot[%
width=4cm,
height=4cm,
scale only axis,
xmin=0,
xmax=20,
xlabel style={font=\color{white!15!black}},
xlabel={\Large Reg.~param.~$\lambda$},
ymin=18,
ymax=22,
ylabel style={font=\color{white!15!black}},
ylabel={\Large ${\sf Wel}(\GSO)$, $f(x)=x$},
axis background/.style={fill=white},
xmajorgrids,
ymajorgrids,
xtick={0,5,10,15,20}, 
legend style={legend cell align=left, align=left, draw=white!15!black}
]
\addplot [color=mycolor1, dashdotted,   line width=2.0pt]
  table[row sep=crcr]{%
0.01	18.5704046360173\\
20	18.5704046360173\\
};

\addplot [color=mycolor2, line width=3.0pt, mark=x, mark options={solid, mycolor2}]
  table[row sep=crcr]{%
0.01	18.5709170158378\\
1	18.5404392529392\\
5	18.8343646767583\\
10	19.2489121363881\\
15	19.5083283161599\\
20	20.8866333881842\\
};

\addplot [color=mycolor4, dashdotted, line width=3.0pt, mark=diamond, mark options={solid, mycolor4}]
  table[row sep=crcr]{%
0.01	18.5703137275574\\
1	18.5821857792905\\
5	18.6035103869766\\
10	18.6223110963369\\
15	18.6333424682045\\
20	18.8012088366425\\
};
\coordinate (top) at (rel axis cs:1,0);
 \hspace{0.5cm}



\coordinate (bot) at (rel axis cs:0,0);
\end{groupplot}

\path (top|-current bounding box.north)--
      coordinate(legendpos)
      (bot|-current bounding box.north);
\matrix[
    matrix of nodes,
    anchor=south,
    draw,
    inner sep=0.1em,
    draw,
    font = \small,
  ]at([yshift=1.75ex,xshift=-23ex]legendpos)
  {\ref{plots:Smooth-GL}& {\Large Smooth-GL \cite{kalofolias2016learn}}&[5pt]
  \ref{plots:bi-app}& {\Large Linear Approx. by \eqref{eq:net-bilevel-app}} \\
  \ref{plots:xiaowen}& {\Large QuadGame GL \cite{leng2020learning}} &[5pt] 
  \ref{plots:spec}& {\Large SpecTemp \cite{segarra2017network}} \\
  \ref{plots:bi-x}& {\Large GLGP $f(x) = x$}&[5pt]
\ref{plots:bi-log}& {\Large GLGP $f(x) = \log(1+x)$}\\};

\end{tikzpicture}%
}\vspace{-.2cm}
\end{center}
\caption{Performance of \eqref{eq:glfp} via TTGD on learning from PA graphs. (Left) AUC performance. (Right) Social welfare.}\vspace{-.3cm}\label{fig:cr_game}
\end{figure}

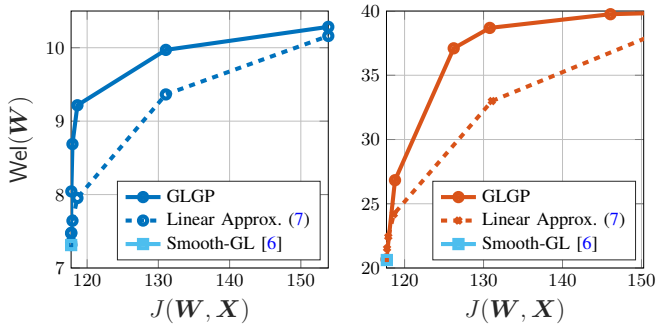
\begin{figure}[t]
\centering
\resizebox{0.99\linewidth}{!}{\definecolor{mycolor1}{rgb}{0.00000,0.44700,0.74100}%
\definecolor{mycolor2}{rgb}{0.85000,0.32500,0.09800}%
\definecolor{mycolor6}{rgb}{0.30100,0.74500,0.93300}%

\begin{tikzpicture}

\begin{groupplot}[group style={group name=myplot,group size=2 by 1}]
\nextgroupplot[%
width=6cm,
height=6cm,
xmin=117.749304433409,
xmax=153.856055193291,
xlabel style={font=\color{white!15!black}},
xlabel={\Large $J( \GSO, {\bm X} )$},
ymin=7,
ymax=10.5,
ylabel style={font=\color{white!15!black}},
ylabel={\Large ${\sf Wel}(\GSO)$},
axis background/.style={fill=white},
xmajorgrids,
ymajorgrids,
legend pos = south east,
legend style={legend cell align=left, align=left, draw=white!15!black}
]
\addplot [color=mycolor1, line width=2.0pt, mark=o, mark options={solid, mycolor1}]
  table[row sep=crcr]{%
117.755985796927	7.31405422792078\\
117.749304433409	7.47184527708778\\
117.782104858476	8.04305794038479\\
117.940064174086	8.68789807765267\\
118.636172592594	9.21653338208747\\
131.052579911335	9.97051561007478\\
153.856055193291	10.284132671904\\
};
\addlegendentry{GLGP}

\addplot [color=mycolor1, dashed, line width=2.0pt, mark=o, mark options={solid, mycolor1}]
  table[row sep=crcr]{%
117.755985796927	7.31230134894373\\
117.749304433409	7.34558697423854\\
117.782104858476	7.47979454160738\\
117.940064174086	7.64365280256893\\
118.636172592594	7.95344765258109\\
131.052579911335	9.36391026366383\\
153.856055193291	10.1600951608408\\
};
\addlegendentry{Linear Approx. \eqref{eq:net-bilevel-app}}

\addplot[
    color=mycolor6,
    line width=2.0pt,
    mark=square,
    mark options={solid, mycolor6}
    ]
    coordinates {
    (117.755985796927,7.31230134894373)
    };
    \addlegendentry{Smooth-GL \cite{kalofolias2016learn}}

\nextgroupplot[%
width=6cm,
height=6cm,
xmin=117.710637041695,
xmax=150.279764294101,
xlabel style={font=\color{white!15!black}},
xlabel={\Large $J( \GSO, {\bm X} )$},
ymin=20,
ymax=40,
ylabel style={font=\color{white!15!black}},
axis background/.style={fill=white},
xmajorgrids,
ymajorgrids,
legend pos = south east,
legend style={legend cell align=left, align=left, draw=white!15!black}
]
\addplot [color=mycolor2, line width=2.0pt, mark=o, mark options={solid, mycolor2}]
  table[row sep=crcr]{%
117.710637041695	20.6718312723392\\
118.788928778172	26.8324380236931\\
126.222203303195	37.0975231907485\\
130.790878925626	38.6854207821264\\
146.073861920048	39.751973328749\\
161.279764294101	39.9989700303178\\
};
\addlegendentry{GLGP}

\addplot [color=mycolor2, dashed, line width=2.0pt, mark=x, mark options={solid, mycolor2}]
  table[row sep=crcr]{%
117.755985796927	20.6200351322813\\
117.749304433409	20.7980319136741\\
117.782104858476	21.5245812471704\\
117.940064174086	22.4298416785578\\
118.636172592594	24.1909329777626\\
131.052579911335	32.9902296181939\\
153.856055193291	38.7289883933107\\
};
\addlegendentry{Linear Approx. \eqref{eq:net-bilevel-app}}

\addplot[
    color=mycolor6,
    line width=2.0pt,
    mark=square,
    mark options={solid, mycolor6}
    ]
    coordinates {
    (117.755985796927,20.6200351322813)
    };
    \addlegendentry{Smooth-GL \cite{kalofolias2016learn}}



\end{groupplot}
\end{tikzpicture}%
}\vspace{-.1cm}
 \caption{Comparing the social welfare ${\sf Wel}(\GSO) = {\bf 1}^\top {\bm y}^{\sf NE}(\GSO)$ against the data fidelity term $J( \GSO, {\bm X})$ for GLGP. (Left) Prior with $f(x) = \log(1+x)$. (Right) Prior with $f(x) = x$.} \vspace{-.4cm} \label{fig:nume-pareto} 
\end{figure}

We benchmark the proposed \eqref{eq:glfp} formulation against two methods: (i) the Smooth-GL method  \cite{kalofolias2016learn} as in \eqref{eq:graph_learn} without any structural nor network games regularization (referred as {\sf Smooth-GL}), (ii) the linear approximation of the bilevel problem \eqref{eq:net-bilevel-app}, where the optimization problem is simplified into a single level problem by approximating \( {\bf 1}^\top \bm{y}^\star (\GSO) \) as \( {\bf 1}^\top \GSO \bm{b} \) (referred as {\sf linear approx.}). For the bilevel optimization problem \eqref{eq:glfp}, the network game is specified with the marginal benefit {${\bm b} \geq {\bm 0}$, ${\bm b} =\max({\bm v}_1,0)$} set as the top eigenvector ${\bm v}_1$ of the Euclidean distance matrix, ${\bm D}$, of the observed graph signals, i.e., $D_{ij} = \| {\bm x}_i^{\rm row} - {\bm x}_j^{\rm row} \|^2$. The marginal benefit vector ${\bm b}$ is further normalized such that ${\bm b}^\top {\bf 1} = 1$. Furthermore, we consider two settings for the interaction function $f(\cdot)$ in \eqref{eq:glfp} wih $f(x) = x, f(x) = \log(1+x)$.
The bilevel optimization problem is tackled using the TTGD algorithm with {step sizes $(\alpha, \gamma) = (0.5,0.003)$ and terminated after $700$ iterations for  $f(x) = x$ and $195$ iterations for $f(x) = \log(1+x)$}.

Fig.~\ref{fig:cr_game} (left) shows the area under ROC curve (AUC) for graph learning performance and the total social welfare of the learnt $\GSO$ against regularization parameter $\lambda$, averaged over 20 Monte-Carlo trials. As $\lambda$ increases, the regularized graph learning objective will become more dependent on the network games prior and tend to learn a graph topology with few hub nodes [cf.~Proposition~\ref{prop:kkt}]. Observe that with $M \ll N$, it is desired to incorporate the network games prior as we observe that tackling \eqref{eq:glfp} yields a solution with better AUC and social welfare than benchmark approaches such as Smooth-GL, GLSigRep. 
Tackling \eqref{eq:glfp} using the prior with the interaction function $f(x) = x$ performs best in this case. Furthermore, we observe from Fig.~\ref{fig:cr_game} (right) that the graph topology found by \eqref{eq:glfp} achieves the highest social welfare.

Our second experiment considers learning the topology of the \texttt{Karate Club} graph, which consists of \( N = 34 \) nodes, with \( M = 50 \) samples of smooth graph signals generated from the Gaussian Markov Random Field (GMRF) model with the precision matrix given by the graph Laplacian \cite{dong2016learning}. Our aim is to showcase the necessity of tackling the bilevel problem \eqref{eq:glfp} using TTGD in lieu of the single-level optimization approximation \eqref{eq:net-bilevel-app}. Fig.~\ref{fig:nume-pareto} shows the \emph{Pareto front} of the bilevel solution and the approximate solution, computed by varying the regularization parameter \( \lambda \) that trades between the smooth-GL objective $J(\GSO, \bm{X})$ and ${\sf Wel}( \GSO )$. As expected, we observe that the TTGD algorithm achieves a better Pareto front than the approximate solution. 


\begin{table}[t]
\centering
\begin{tabular}{lccc}
\toprule
 & \textbf{Maximum} & \textbf{Average} & \textbf{Minimum} \\ 
  \midrule
GLGP ($f(x) = x$)  &   0.6345              &   0.5788              &  0.5271               \\ 
GLGP ($f(x) = \log(1+x)$)  &    {\bf 0.6570}             & {\bf  0.5937}              &   \bfseries 0.5490              \\ 
\midrule 
Linear Approx. by \eqref{eq:net-bilevel-app}  &   0.6429              &   0.5600               &    0.5087              \\ 
Smooth-GL \cite{kalofolias2016learn} &     0.5075             &    0.4888             &    0.4777             \\ 
QuadGame-GL\cite{leng2020learning} &    0.5735             &    0.5419             &    0.5120             \\ 
SpectTemp \cite{segarra2017network} & 0.5760 & 0.5187 & 0.4463\\
  \bottomrule
\end{tabular}
\caption{Comparing the AUC performance of the graph learnt from the {\tt IndianVillage} data \cite{banerjee2013diffusion}.}\vspace{-.1cm} \label{tab:auroc}
\end{table}
           
\begin{table}[t]
\centering
\begin{tabular}{lccc}
\toprule
${\sf Wel}( \hat{\GSO} ) - {\sf Wel}( {\GSO}^{\sf true} )$\footnote{$f(\cdot)$ specifies the interaction function that the welfare is evaluated on.} & \textbf{Maximum} & \textbf{Average} & \textbf{Minimum} \\ 
\midrule
GLGP ($f(x) = x$)         &   \bfseries 4.2754              &   {\bf 3.2077 }              &  \bfseries 2.1693              \\ 
Linear Approx.~by \eqref{eq:net-bilevel-app}  &    -0.4790          &    -1.4241             &     -2.3927            \\ 
Smooth-GL \cite{kalofolias2016learn} &     -2.3288            &     -5.8141            &     -8.6115            \\ 
QuadGame-GL \cite{leng2020learning} &     -1.9386            &     -3.5631            &     -5.5949            \\ 
SpectTemp \cite{segarra2017network}  &     0.2653            &     -2.8122            &     -6.6532            \\ 
\midrule
GLGP ($f(x) = \log(1+x)$)         &   \bfseries 1.4464              &  {\bf 1.1892 }             &  \bfseries 0.8461              \\ 
Linear Approx.~by \eqref{eq:net-bilevel-app} &   -0.0676              &    -0.6144              &  -1.1429               \\ 
Smooth-GL \cite{kalofolias2016learn}  &   -0.5440              &   -2.8292              &  -5.0419               \\ 
QuadGame-GL \cite{leng2020learning}  &   -0.8206              &   -1.6814              &  -3.3074                \\ 
SpectTemp \cite{segarra2017network} &     0.1453            &     -1.2226            &     -3.6324            \\
  \bottomrule
\end{tabular}
\caption{Comparing the gain in social welfare of the graph learnt from the {\tt IndianVillage} data \cite{banerjee2013diffusion}. Top (resp.~bottom) rows are evaluated using $f(x)=x$ (resp.~$f(x)=\log(1+x)$).}\vspace{-.4cm} \label{tab:wel}
\end{table}

\vspace{-0.0cm}
\subsection{Case Studies with Real Data}
Our last set of experiments consider learning the graph topology from a set of real data taken from Indian villages \cite{banerjee2013diffusion}. The dataset consists of survey data from $40$ villages, where the village networks have sizes ranging from $N=77$ to $N=330$ agents, and $M=16$ samples of graph signals are observed for each network. 

We set the parameters in \eqref{eq:glfp} with \(\beta = 1\), \({\bm b} = {\tt h} + 1\), where \( {\tt h}_i \in \{0,1\} \) indicates if the agent is a potential microfinance client (cf.~`{\tt hhSurveyed}' in the dataset) and ${\bm b}$ is normalized such that ${\bm b}^\top {\bm 1}=1$. Additionally, we set \(\lambda = 100\) and apply the TTGD algorithm with a stepsize of \((\alpha, \gamma) = (0.5, 0.003)\) for $760$ (with $f(x)=x$) and $320$  (with $f(x)=\log(x+1)$) iterations. 

Table~\ref{tab:auroc} reports the maximum/average/minimum AUC performances among the graph topologies learnt under various settings and algorithms, compared to the ground truth in \cite{banerjee2013diffusion}. We notice that GLGP consistently estimates the graph topologies accurately when used with the interaction function $f(x) = \log(1+x)$. 
Similarly, Table~\ref{tab:wel} reports the maximum/average/minimum gain in social welfare compared to the ground truth, for the graph topologies learnt. As expected, GLGP finds graph topologies that achieve better performance in the social welfare.


\section{Conclusions}
We have proposed a novel  graph topology learning formulation that incorporates prior knowledge of dynamics over the networks. We study a prior induced by general linear quadratic games and formulate a bilevel program that can be tackled by a TTGD algorithm. We envisage that the proposed GLGP framework can inspire new graph learning formulations using domain knowledge from GSP, network dynamics, and game theory. 

{
\bibliographystyle{IEEEtran}
\bibliography{ref}
}

\newpage 

\allowdisplaybreaks
\appendices

\section{Proof of Proposition~\ref{prop:approx}} \label{pf:approx}
\begin{proof} 
To facilitate our proof, we denote $\bar{\cal S}$ as the set of $(\GSO, {\bm y})$ feasible to \eqref{eq:glfp}.
We first note that for any $\GSO \in {\cal S}$, it holds $\GSO {\bf f}( {\bm y} ) + {\bm b} \geq {\bm 0}$ and the equilibrium constraint is equivalent to
\beq 
{\bm y} = \GSO {\bf f}( {\bm y} ) + {\bm b}
\eeq 
Since $f(x) \leq x$ and $\GSO \in {\cal S}$ is non-negative, the above constraint can be \emph{relaxed} to 
\beq 
{\bm y} \leq \GSO {\bm y} + {\bm b} \Longrightarrow {\bm y} \leq ( {\bm I} - \GSO )^{-1} {\bm b}
\eeq 
The above can be further relaxed to:
\[
 \begin{aligned}
   {\bm y} \leq ({{\bm I}-\GSO})^{-1} {\bm b} & \leq ({{\bm I}+\GSO}){\bm b}+ \| {\bm b} \|_\infty (c^2+c^3+\ldots){\bm 1} \\
  &= ({{\bm I}+\GSO}){\bm b}+ {(1-c)^{-1} c^2 \| {\bm b} \|_\infty} {\bm 1}
 \end{aligned}
\]
Denoting $\bar{\bm c} = {(1-c)^{-1} c^2 \| {\bm b} \|_\infty} {\bf 1} + {\bm b}$, we have 
\beq 
\bar{\cal S} \subseteq \bar{\cal S}_{\sf relax} := 
\{ ( \GSO, {\bm y} ) : \GSO \in {\cal S},~{\bm y} \leq \GSO {\bm b}+ \bar{\bm c} \}.
\eeq 
Subsequently, the following problem is a relaxation of \eqref{eq:glfp}:
\beq \textstyle \label{eq:net-bilevel-app-cons}
\min_{ (\GSO, {\bm y}) \in \bar{\cal S}_{\sf relax} } J( \GSO; {\bm X} ) - \lambda {\bf 1}^\top {\bm y}
\eeq 
Applying the optimality principle for ${\bm y}$ shows that \eqref{eq:net-bilevel-app-cons} is equivalent to \eqref{eq:net-bilevel-app}. 
\end{proof}

\section{Proof of Proposition~\ref{prop:kkt}} \label{pf:kkt}

\begin{proof}
To simplify notations, we define the Euclidean distance matrix ${\bm D} \in \RR^{N \times N}$ with ${ D}_{ij} = \frac{\| {\bm x}_i^{\text{row}} - {\bm x}_j^{\text{row}} \|^2}{2M}$.
Introducing the dual variables $\lambda \in \RR, {\bm \eta} \in \mathbb{R}^N$, ${\bm h} \in \mathbb{R}^N$, $\boldsymbol{\mu} \in \mathbb{R}_+^{N \times N}$, the  Lagrangian function of the approximate GLGP problem \eqref{eq:net-bilevel-app} is given by:
\beq \notag
\begin{aligned}
& {\cal L}( \GSO, \lambda, \bm{\eta}, \bm{\mu}, {\bm h} ) \\
&= \langle \GSO, {\bm D} \rangle + \beta \| \GSO \|_{\text{F}}^2 - \lambda {\bm 1}^\top \GSO{\bm b} \\
  & \textstyle \quad - {\bm \eta}^\top(\GSO{\bm 1} - c{\bm 1}) - \langle \boldsymbol{\mu}, \GSO \rangle-\sum_{i=1}^N h_i{\bm e}_i^\top \GSO {\bm e}_i.
\end{aligned}
\eeq
We consider the first order necessary condition, i.e., setting the gradient of ${\cal L}(\cdot)$ w.r.t.~$\GSO$ to zero. This yields:
\[
{\bm D} + 2\beta \GSO - \lambda {\bm 1}{\bm b}^\top - \boldsymbol{\mu} - {\bm \eta}{\bm 1}^\top- {\rm Diag}( {\bm h} )
= {\bm 0}.
\]
On the other hand, the complementary slackness condition yields $W_{ij} \mu_{ij} = 0$. Together with the primal feasibility $\GSO \geq {\bm 0}$, we obtain the following conditions for $i \neq j$:

\vspace{.1cm}
\noindent {\bf Case 1.} If $W_{ij} > 0$, we must have $\mu_{ij} = 0$ and thus
\[
    W_{ij} = -\frac{{ D}_{ij}}{2 \beta} + \frac{\lambda { b}_j}{2 \beta} + \frac{\eta_i}{2 \beta},
\]
\noindent {\bf Case 2.} If $W_{ij} = 0$, with $\mu_{ij} \geq 0$, we get 
\[
    -\frac{{ D}_{ij}}{2 \beta} + \frac{\lambda { b}_j}{2 \beta} + \frac{\eta_i}{2 \beta} \leq 0.
\]
Therefore, for $i \neq j$, the optimal solution for $\GSO$ is:
\[
W^*_{ij} = \frac{1}{2\beta} \max\left( \eta_i - \frac{\| {\bm x}_i^{\text{row}} - {\bm x}_j^{\text{row}} \|^2}{2M} + \lambda b_j, 0 \right).
\]
This concludes our proof.
\end{proof}

\section{Proof of Theorem~\ref{thm:ttgd}} \label{app:ttgd}
The precise statement of the theorem is repeated as follows:
\begin{Theorem*} \notag
Assume H\ref{assu:lipsf}, then with the step sizes $\alpha = \frac{1-c}{(1+c)^2}$ and $\gamma \leq \min\{\frac{3}{4L_{\ell}}, \frac{1-c}{4LL_y}\alpha\}$, it holds for any $K \geq 1$,
\beq \notag
\min_{k=1,...,K} \| \gamma^{-1} ( \GSO^k - {\sf Proj}_{\cal S} ( \GSO^k - \gamma \grd \ell( \GSO^k ) ) ) \|^2 = {\cal O}(K^{-1}).
\eeq 
Here, $L_\ell = \left(\frac{\sqrt{N}\lambda}{1-c} + \frac{\lambda N\|\bm{b}\|_{\infty}c}{(1-c)^3}\right)\frac{\sqrt{N}\|\bm{b}\|_{\infty}}{(1-c)^2} + 2\beta + \frac{\lambda N\|\bm{b}\|_{\infty}}{(1-c)^3}$ is the Lipschitz constant of $\grd\ell$, $L_y = \frac{\sqrt{N}\|\bm{b}\|_{\infty}}{(1-c)^2}$ is the Lipschitz constant of $\bm{y}^{\sf NE}(\cdot)$ and $L = \frac{\sqrt{N}\lambda}{1-c} + \frac{\lambda \sqrt{N}c}{(1-c)^2}\left(\frac{\sqrt{N}\|\bm{b}\|_{\infty}}{1-c} + \sqrt{\left(1-\frac{\alpha(1-c)}{2}\right)\|\bm{y}^{1} - \bar{\bm{y}}^{0}\|^2 + 2cL_y^2(\frac{2}{(1-c)\alpha}-1)}\right)$.
\end{Theorem*}

For the sake of brevity, we define a mapping ${\sf G_{\GSO}} : \mathbb{R}^n_+ \rightarrow \mathbb{R}^n_+$ for any feasible $\GSO$ as $\sf G_{\GSO}(\bm{y}) := \bm{y} - {\sf T}(\bm{y};\GSO)$. The properties of $\sf G$ are discussed in the follow lemma. With a slight modification to Lemma 4 in \cite{Candogan}, it can be shown that $\rho(\GSO) \leq c$ for all feasible $\GSO \in \mathcal{S}$. 
\begin{Lemma}\label{lemma:property-G}
    Assume $\bm{b} \geq \bm{0}$ and $c < 1$, for any $\GSO \in \mathcal{S}$, we have
    \begin{enumerate}
        \item $\|{\sf G_{\GSO}}(\bm{y}_1) - {\sf G_{\GSO}} (\bm{y}_2)\| \leq (1+c)\|\bm{y}_1 - \bm{y}_2\|, \forall \bm{y}_1, \bm{y}_2 \in \mathbb{R}_+^n$;
        \item $\langle{\sf G_{\GSO}}(\bm{y}_1) - {\sf G_{\GSO}} (\bm{y}_2), \bm{y}_1 - \bm{y}_2\rangle \geq (1-c)\|\bm{y}_1 - \bm{y}_2\|^2, \forall \bm{y}_1, \bm{y}_2 \in \mathbb{R}_+^n$.
    \end{enumerate}
\end{Lemma}
\begin{proof}
    Since $\bm{b} \geq {0}$, $\sf T(\bm{y}, \GSO) = \bm{b} + \GSO f(\bm{y})$. Hence, 
    \[
    {\sf G_{\GSO}}(\bm{y}_1) - {\sf G_{\GSO}}(\bm{y}_2) = (\bm{y}_1 - \bm{y}_2) - \GSO(f(\bm{y}_1) - f(\bm{y}_2)).
    \]
    For the first argument, 
    \begin{equation*}
        \begin{aligned}
           \| {\sf G}_{\GSO}(\bm{y}_1) - {\sf G}_{\GSO}(\bm{y}_2)\| &\leq \|\bm{y}_1 - \bm{y}_2\| + \rho(\GSO)\|f(\bm{y}_1) - f(\bm{y}_2)\| \\
           &\leq (1+c)\|\bm{y}_1 - \bm{y}_2\|
        \end{aligned}
    \end{equation*}
    The last inequality is due to $\rho(\GSO) \leq c$ and $f$ is $1$-Lipschitz.
    The second argument can be derived by similar steps.
\end{proof}
The following series of lemmas characterize the properties of $\ell$ and $\widehat{\grd} \Phi$.
We recall that 
\begin{equation*}
    \begin{aligned}
        &\widehat{\grd} \Phi( \GSO ,\bm y ) \\
        = &\grd_{\GSO} \Phi( \GSO; \bm y ) - ( {\rm J}_{\GSO} {\sf T}( \bm y; \GSO) )^\top ( {\rm J}_{\bm y} {\sf T}( \bm y; \GSO) - {\bm I}_N )^{-\top} \grd_{\bm y} \Phi( \GSO; \bm y )
    \end{aligned}
\end{equation*}
We begin by simplifying the expression of $\widehat{\grd} \Phi( \GSO ,\bm y )$.
\begin{Lemma}\label{lemma: grdPhi-expression}
    Assume $\bm{b} \geq \bm{0}$, 
    \begin{equation*}
        \begin{aligned}
            \widehat{\grd} \Phi( \GSO ,\bm y ) = \grd_\GSO J(\GSO) - \lambda (f(\bm y)\otimes \bm{I}_N)(\bm{I}_N - \GSO {\sf Diag}(f'(\bm y)))^{-\top}\bm{1}
        \end{aligned}
    \end{equation*}
\end{Lemma}
\begin{proof}
    Since ${\sf T}( \bm y; \GSO) ) = \bm{b} + \GSO f(\bm{y})$ when $\bm{b} \geq \bm{0}$, 
    \begin{equation*}
        \begin{aligned}
            &{\rm J}_{\GSO} {\sf T}( \bm y; \GSO) = f(\bm{y})^\top\otimes \bm{I}_N \\
            &{\rm J}_{\bm{y}} {\sf T}( \bm y; \GSO) = \GSO {\sf Diag}(f'(\bm y)).
        \end{aligned}
    \end{equation*}
    Since $\Phi( \GSO; \bm y ) = J(\GSO) - \lambda \bm{1}^\top\bm{y}$, 
        \begin{equation*}
        \begin{aligned}
&\grd_\GSO\Phi( \GSO; \bm y ) = \grd J(\GSO) \\
&\grd_{\bm y}\Phi( \GSO; \bm y ) = -\lambda \bm{1}.
        \end{aligned}
    \end{equation*}
    Combining them gives the results.
\end{proof}

\begin{Lemma}\label{lemma:property-yne}
    Under H\ref{assu:lipsf}, ${\bm y}^{\sf NE} ( \GSO )$ is $L_y$-Lipschitz w.r.t.~$\GSO$, where 
    \[ L_y = \frac{\sqrt{N}\|\bm{b}\|_{\infty}}{(1-c)^2}.
    \]
\end{Lemma}
\begin{proof}
    Since $f$ is $1$-Lipschitz, for any $\bm{y}\geq \bm{0}$, we have
    \begin{equation*}
        f(\bm{y}) \leq \|\bm{y}\|_{\infty}\bm{1}.
    \end{equation*}
    From the definition of $\bm{y}^{\sf NE}$, we can infer that
    \begin{equation}
        \begin{aligned}
            \|\bm{y}^{\sf NE}\|_{\infty} &\leq \|\bm{b} + \GSO f(\bm{y}^{\sf NE})\|_{\infty} \\
            &\leq \|\bm{b}\|_{\infty} + \|\GSO\|_1\|f(\bm{y}^{\sf NE})\|_{\infty} \\
            &\leq \|\bm{b}\|_{\infty} + c\|\bm{y}^{\sf NE}\|_{\infty}. \\
            \Rightarrow \|\bm{y}^{\sf NE}\|_{\infty} &\leq \frac{1}{1-c} \|\bm{b}\|_{\infty}. \\
            \Rightarrow \|f(\bm{y}^{\sf NE})\|_{\infty} &\leq \frac{1}{(1-c)} \|\bm{b}\|_{\infty}.
        \end{aligned}
    \end{equation}
    For simplicity, denote ${\bm y}_1 = {\bm y}^{\sf NE}( \GSO_1 )$, ${\bm y}_2 = {\bm y}^{\sf NE}( \GSO_2 )$. We have
    \begin{equation*}
        \begin{aligned}
            \|\bm{y}_1 - \bm{y}_2\|_2 &= \|(\bm{b} + \GSO_1f(\bm{y}_1))^+ - (\bm{b} + \GSO_2f(\bm{y}_2))^+\|_2 \\
            &\leq  \|(\bm{b} + \GSO_1f(\bm{y}_1)) - (\bm{b} + \GSO_2f(\bm{y}_2))\|_2 \\
            &\leq \|\GSO_1\|\|\bm{y}_1 - \bm{y}_2\|_2 + \|f(\bm{y}_2)\|_2\|\GSO_1 - \GSO_2\| \\
            &\leq c\|\bm{y}_1 - \bm{y}_2\|_2 + \sqrt{N}\|f(\bm{y}_2)\|_{\infty}\|\GSO_1 - \GSO_2\| \\
            &\leq c\|\bm{y}_1 - \bm{y}_2\|_2 + \frac{\sqrt{N}\|\bm{b}\|_{\infty}}{1-c}\|\GSO_1 - \GSO_2\|. \\
            \Rightarrow \|\bm{y}_1 - \bm{y}_2\|_2 &\leq \frac{\sqrt{N}\|\bm{b}\|_{\infty}}{(1-c)^2}\|\GSO_1 - \GSO_2\|.
        \end{aligned}
    \end{equation*}
    This concludes the proof.
\end{proof}

\begin{Lemma}\label{lemma: property-hatphi}
For any $\GSO \in \mathcal{S}$ and $B > 0$, the gradient map $\widehat{\grd}\Phi(\GSO, {\bm y})$ is $L$-Lipschitz w.r.t. ${\bm y}$ for all ${\bm y} $ satisfying $\|\bm{y}\|\leq B$, where 
\[ 
L =\frac{\sqrt{N}\lambda}{1-c} + \frac{\lambda \sqrt{N}Bc}{(1-c)^2}.
\]
\end{Lemma}
\begin{proof}
    From Lemma \ref{lemma: grdPhi-expression}, it remains to consider for any $\bm{y}_1$, $\bm{y}_2$ the following chain:
    \begin{equation} \notag
    \begin{aligned}
        &\|\widehat{\grd}\Phi(\GSO, \bm{y}_1) - \widehat{\grd}\Phi(\GSO, \bm{y}_2)\|_2 \\
        &\leq \lambda \|\left((f(\bm{y}_1) - f(\bm{y}_2)) \otimes \bm{I}_N\right)(\bm{I}_N - \GSO{\sf Diag}(f'(\bm{y}_1)))^{-\top}\bm{1}\|_2 \\
        &+ \lambda \|\left(f(\bm{y}_2) \otimes \bm{I}_N\right)\left((\bm{I}_N - \GSO{\sf Diag}(f'(\bm{y}_1)))^{-\top}\right. \\
        &\left. - (\bm{I}_N - \GSO{\sf Diag}(f'(\bm{y}_2)))^{-\top}\right)\bm{1}\|_2 \\
        & \leq \lambda\sqrt{N} \|\bm{y}_1 - \bm{y}_2\|_2\|(\bm{I}_N - \GSO{\sf Diag}(f'(\bm{y}_1)))^{-\top}\|_2 \\
        &+ \lambda\sqrt{N}\|f(\bm{y}_2)\|\\
        &* \|(\bm{I}_N - \GSO{\sf Diag}(f'(\bm{y}_1)))^{-\top}\|_2\|(\bm{I}_N - \GSO{\sf Diag}(f'(\bm{y}_2)))^{-\top}\|_2 \\
        &* \|({\sf Diag}(f'(\bm{y}_2)) - {\sf Diag}(f'(\bm{y}_1)))\GSO\| \\
        &\leq \frac{\sqrt{N}\lambda}{1-c} \|\bm{y}_1 - \bm{y}_2\|_2 + \frac{\lambda \sqrt{N}Bc}{(1-c)^2} \|\bm{y}_1 - \bm{y}_2\|_2 \\
        & = \left(\frac{\sqrt{N}\lambda}{1-c} + \frac{\lambda \sqrt{N}Bc}{(1-c)^2}\right)\|\bm{y}_1 - \bm{y}_2\|_2.
    \end{aligned}
    \end{equation}
    This yields the required Lipschitz constant.
\end{proof}

This lemma shows that $\widehat{\grd}\Phi(\GSO, \cdot)$ is Lipschitz continuous on any bounded set. The Lipschitz constant is independent of $\GSO$ but relies on the radius $B$. Later we will show that $\{\bm{y}^k\}$ is bounded in TTGD in Proposition \ref{prop:inequality-recursion} and hence this lemma is applicable.
 
\begin{Lemma}\label{lemma: property-l}
(Property of $\ell$) Under assumption H\ref{assu:lipsf}, $\grd{\ell}$ is $L_{\ell}$-Lipschitz continuous, where 
\[ L_\ell = \left(\frac{\sqrt{N}\lambda}{1-c} + \frac{\lambda N\|\bm{b}\|_{\infty}c}{(1-c)^3}\right)\frac{\sqrt{N}\|\bm{b}\|_{\infty}}{(1-c)^2} + 2\beta + \frac{\lambda N\|\bm{b}\|_{\infty}}{(1-c)^3}.
\]
\end{Lemma}
\begin{proof}
    Recall that $\grd \ell( \GSO ) = \widehat{\grd} \Phi( \GSO , \bar{\bm y} )$. Hence, for all $\GSO_1, \GSO_2 \in \mathcal{S}$, 
    \begin{equation*}
        \begin{aligned}
            &\|\grd \ell( \GSO_1 ) - \grd \ell( \GSO_2 )\| = \|\widehat{\grd} \Phi( \GSO_1 , \bar{\bm y}_1 ) - \widehat{\grd} \Phi( \GSO_2 , \bar{\bm y}_2 )\| \\ 
            \leq &\|\widehat{\grd} \Phi( \GSO_1 , \bar{\bm y}_1 ) - \widehat{\grd} \Phi( \GSO_2 , \bar{\bm y}_1 )\| + \|\widehat{\grd} \Phi( \GSO_2 , \bar{\bm y}_1 ) - \widehat{\grd} \Phi( \GSO_2 , \bar{\bm y}_2 )\|
        \end{aligned}
    \end{equation*}
    The proof of Lemma \ref{lemma:property-yne} shows that for any NE $\bar{\bm{y}}$, it is bounded by $\frac{\sqrt{N}\|\bm{b}\|_{\infty}}{1-c}$. Hence, we can choose $B = \frac{\sqrt{N}\|\bm{b}\|_{\infty}}{1-c}$ in Lemma \ref{lemma: property-hatphi}, which together with the result from Lemma \ref{lemma:property-yne} give the upper bound on the second term:  \[\|\widehat{\grd} \Phi( \GSO_2 , \bar{\bm y}_1 ) - \widehat{\grd} \Phi( \GSO_2 , \bar{\bm y}_2 )\| \leq L_1\|\GSO_1 - \GSO_2\|,\]
    where $L_1 = \left(\frac{\sqrt{N}\lambda}{1-c} + \frac{\lambda N\|\bm{b}\|_{\infty}c}{(1-c)^3}\right)\frac{\sqrt{N}\|\bm{b}\|_{\infty}}{(1-c)^2}$. So it remains to bound $\|\widehat{\grd} \Phi( \GSO_1 , \bar{\bm y}_1 ) - \widehat{\grd} \Phi( \GSO_2 , \bar{\bm y}_1 )\|$ with $\|\GSO_1 - \GSO_2\|$.
    
    Since $\grd \Phi( \GSO , \bar{\bm y} ) = \grd J(\GSO)$ and $\grd J(\GSO)$ is $2\beta$-Lipschitz, our remaining task is to show that the map 
    \[
    {P(\GSO)} := - ( {\rm J}_{\GSO} {\sf T}( \bar{\bm y}_1; \GSO) )^\top ( {\rm J}_{\bm y} {\sf T}( \bar{\bm y}_1; \GSO) - {\bm I}_N )^{-\top} \grd_{\bm y} \Phi(\GSO; \bar{\bm y}_1 )
    \] 
    is Lipschitz continuous over ${\cal S}$.
    Note that for any $\GSO_1, \GSO_2 \in \mathcal{S}$, denote $\widehat{\GSO}_1 = {\sf Diag}(f'(\bar{\bm{y}}_1))\GSO_1$ and $\widehat{\GSO}_2 = {\sf Diag}(f'(\bar{\bm{y}}_1))\GSO_2$, 
    \begin{equation*}
        \begin{aligned}
            &\|P(\GSO_1) - P(\GSO_2)\| \\
            = &\lambda\|(f(\bar{\bm{y}}_1) \otimes \bm{I}_N)[(\bm{I}_N - \widehat{\GSO}_1)^{-1} - (\bm{I}_N - \widehat{\GSO}_2)^{-1}]\bm{1}\| \\
            \leq &\lambda \sqrt{N}\|\bar{\bm{y}}_1\|\|(\bm{I}_N - \widehat{\GSO}_1)^{-1} - (\bm{I}_N - \widehat{\GSO}_2)^{-1}\| \\
            = &\lambda \sqrt{N}\|\bar{\bm{y}}_1\|\|(\bm{I}_N - \widehat{\GSO}_1)^{-1}[\widehat{\GSO}_1 - \widehat{\GSO}_2](\bm{I}_N - \widehat{\GSO}_2)^{-1}\| \\
            \leq &\frac{\lambda \sqrt{N}\|\bar{\bm{y}}_1\|}{(1-c)^2}\|\widehat{\GSO}_1 - \widehat{\GSO}_2\| \\
            \leq &\frac{\lambda \sqrt{N}\|\bar{\bm{y}}_1\|}{(1-c)^2}\|\GSO_1 - \GSO_2\|.
        \end{aligned}
    \end{equation*}
    The proof of Lemma \ref{lemma:property-yne} implies that if $\bm{y}$ is a NE, it is bounded by $\frac{\sqrt{N}}{1-c}\|\bm{b}\|_{\infty}$. Hence, 
    \[\|P(\GSO_1) - P(\GSO_2)\| \leq \frac{\lambda N\|\bm{b}\|_{\infty}}{(1-c)^3}\|\GSO_1 - \GSO_2\|.\]
    The proof is completed.
\end{proof}
\noindent Notice that the above lemma implies 
\begin{enumerate}
        \item $\ell$ is $L_{\ell}$-smooth, e.g. $\forall \GSO_1, \GSO_2 \in \mathcal{S}$
        \[l(\GSO_1)-l(\GSO_2)\leq \langle\nabla l(\GSO_2),\GSO_1-\GSO_2\rangle +(L_{\ell}/2)\|\GSO_1-\GSO_2\|^{2}.\]
        \item   $\ell$ is  $L_{\ell}$- weakly convex, e.g. $\forall \GSO_1, \GSO_2 \in \mathcal{S}$ 
        \[l(\GSO_1)-l(\GSO_2)\geq \langle\nabla l(\GSO_2),\GSO_1-\GSO_2\rangle -(L_{\ell}/2)\|\GSO_1-\GSO_2\|^{2};\]
\end{enumerate}

We also need the follow propositions to assist the proof of Theorem~\ref{thm:ttgd}.
 The following proposition is modified from Lemma 3.6 in \cite{hong2023two}.
 \begin{Prop}\label{prop:inequality-general}
         Consider nonnegative sequences $\{\Omega_k\}, \{\Gamma_k\}, \{\Theta_k\}$. Let $c_0, c_1, d_0, d_1 > 0$, if 
    \beq
    \begin{aligned}
        &\Omega_{k+1} \leq \Omega_{k} - c_0\Theta_{k+1} + c_1\Gamma_{k+1}, \\
        &\Gamma_{k+1} \leq (1-d_0)\Gamma_k + d_1\Theta_k,
    \end{aligned}
    \eeq
    When $\frac{c_0}{c_1} >\frac{d_1}{d_0}$, we have
    \beq
    \begin{aligned}
        &\frac{1}{K}\sum_{k = 1}^K\Theta_k \leq \frac{\Omega_0 + \frac{c_1}{d_0}(\Gamma^0 + d_1\Theta_0)}{(c_0 - c_1d_1/d_0)K}, \\
        &\frac{1}{K}\sum_{k = 1}^K\Gamma_k \leq \frac{\Gamma_0 + d_1\Theta_0 + \frac{d_1}{c_0}\Omega_0}{(d_0 - c_1d_1/c_0)K},
    \end{aligned}
    \eeq
 \end{Prop}
\begin{Prop}\label{prop:inequality-recursion}
    If the lower level stepsize $\alpha$ satisfies $\alpha\leq\frac{1-c}{(1+c)^{2}}$, we have, 
    \begin{enumerate}
        \item     \begin{equation}\label{recursion-y}
        \begin{aligned}
            \|\bm{y}^{k+1}-\bar{\bm{y}}^k\|^{2}\leq &\left(1 - \frac{\alpha(1-c)}{2}\right)\|\bm{y}^k - \bar{\bm{y}}^{k-1}\|^2 \\
            &+ L_y^2\left(\frac{2}{(1-c)\alpha} - 1\right)\|\GSO^{k-1} - \GSO^k\|^2
        \end{aligned}
    \end{equation}
    \item     \begin{equation}
        \begin{aligned}
            l(\GSO^{k+1})-l(\GSO^{k}) \leq& -(\frac{1}{2\gamma}-\frac{L_{\ell}}{2})\|\GSO^{k+1}-\GSO^{k}\|^{2} \\
            &+ \frac{\gamma}{2}L^{2}\|\bm{y}^{k+1}-\bar{\bm{y}}^k\|^{2}
        \end{aligned}
    \end{equation}
    \end{enumerate}
\end{Prop}
\begin{proof}
Since $\bm{y}^{k+1} = \bm{y}^{k} - \alpha G_{\GSO^k}(\bm{y}^k)$, we have
\begin{equation*}
    \begin{aligned}
        \|\bm{y}^{k+1} - \bar{\bm{y}}^k\|^2 = &\|\bm{y}^k - \bar{\bm{y}}^k - \alpha G_{\GSO^k}(\bm{y}^k)\|^2 \\
        = &\|\bm{y}^k - \bar{\bm{y}}^k - \alpha \left(G_{\GSO^k}(\bm{y}^k) - G_{\GSO^k}(\bar{\bm{y}}^k\right)\|^2 \\
        = &\|\bm{y}^k - \bar{\bm{y}}^k\|^2 + \alpha^2\|G_{\GSO^k}(\bm{y}^k) - G_{\GSO^k}(\bar{\bm{y}}^k)\|^2 \\
        &- 2\alpha\langle \bm{y}^k - \bar{\bm{y}}^k, G_{\GSO^k}(\bm{y}^k) - G_{\GSO^k}(\bar{\bm{y}}^k\rangle \\
         \leq &(1 + \alpha^2L_g^2)\|\bm{y}^k - \bar{\bm{y}}^k\|^2 \\
         &- 2\alpha\langle \bm{y}^k - \bar{\bm{y}}^k, G_{\GSO^k}(\bm{y}^k) - G_{\GSO^k}(\bar{\bm{y}}^k)\rangle \\
         \leq &(1 + \alpha^2L_g^2 - 2\alpha\mu_g)\|\bm{y}^k - \bar{\bm{y}}^k\|^2.
    \end{aligned}
\end{equation*}
Here, we define $L_g := 1+c$ and $\mu_g := 1-c$. The first inequality is from the Lipschitz continuity of $G_{\GSO^k}(\cdot)$ (Lemma \ref{lemma:property-G}) and the second inequality is from Lemma \ref{lemma:property-G}.

When $\alpha \leq \frac{\mu_g}{L_g^2}$, we can get
\begin{equation*}
    \begin{aligned}
       &(1 + \alpha^2L_g^2 - 2\alpha\mu_g)\|\bm{y}^k - \bar{\bm{y}}^k\|^2 \leq (1 - \alpha\mu_g)\|\bm{y}^k - \bar{\bm{y}}^k\|^2 \\
        = &(1 - \alpha\mu_g)\|\bm{y}^k - \bar{\bm{y}}^{k-1} + \left(\bar{\bm{y}}^{k-1} - \bar{\bm{y}}^k\right)\|^2 \\
       &\leq (1 - \alpha\mu_g)\left[ (1+z)\|\bm{y}^k - \bar{\bm{y}}^{k-1}\|^2 + (1+\frac{1}{z})\|\bar{\bm{y}}^{k-1} - \bar{\bm{y}}^k\|^2\right]
    \end{aligned}
\end{equation*}
From the Lipschitz continuity of ${\bm y}^{\sf NE} ( \GSO^k )$ (Lemma \ref{lemma:property-yne}), we have 
\begin{equation*}
    \begin{aligned}
        &(1 - \alpha\mu_g)\left[ (1+z)\|\bm{y}^k - \bar{\bm{y}}^{k-1}\|^2 + (1+\frac{1}{z})\|\bar{\bm{y}}^{k-1} - \bar{\bm{y}}^k\|^2\right] \\
        \leq &(1 - \alpha\mu_g) (1+z)\|\bm{y}^k - \bar{\bm{y}}^{k-1}\|^2 \\
        &+ L_y^2(1 - \alpha\mu_g)(1+\frac{1}{z})\|\GSO^{k-1} - \GSO^k\|^2 \\
        \leq &(1 - \alpha\mu_g) (1+z)\|\bm{y}^k - \bar{\bm{y}}^{k-1}\|^2 + L_y^2(1+\frac{1}{z})\|\GSO^{k-1} - \GSO^k\|^2.
    \end{aligned}
\end{equation*}
The last inequality is from the fact that if $\alpha \leq \frac{\mu_g}{L_g^2}$, $\alpha\mu_g < 1$. Set $z = \frac{\mu_g\alpha}{2(1-\mu_g\alpha)}$. We have
\begin{equation*}
    \begin{aligned}
        &(1 - \alpha\mu_g) (1+z)\|\bm{y}^k - \bar{\bm{y}}^{k-1}\|^2 + L_y^2(1+\frac{1}{z})\|\GSO^{k-1} - \GSO^k\|^2 \\
        = &(1 - \alpha\mu_g/2)\|\bm{y}^k - \bar{\bm{y}}^{k-1}\|^2 + L_y^2(\frac{2}{\mu_g\alpha} - 1)\|\GSO^{k-1} - \GSO^k\|^2.
    \end{aligned}
\end{equation*}
This gives the first part result. A direct corollary of this result is that $\bm{y}^k$ is bounded. To see this, note that since $\|\GSO^k - \GSO^{k-1}\| \leq 2c$, \eqref{recursion-y} implies
\begin{equation*}
    \begin{aligned}
        &\|\bm{y}^{k+1} - \bar{\bm{y}}^k\|^2 - \frac{4cL_y^2}{\alpha(1-c)}(\frac{2}{(1-c)\alpha}-1) \\
        \leq &\left(1-\frac{\alpha(1-c)}{2}\right)\left(\|\bm{y}^{k} - \bar{\bm{y}}^{k-1}\|^2 - \frac{4cL_y^2}{\alpha(1-c)}(\frac{2}{(1-c)\alpha}-1)\right)
    \end{aligned}
\end{equation*}
Hence, 
\begin{equation*}
    \begin{aligned}
       |\bm{y}^{k+1} - \bar{\bm{y}}^k\|^2 \leq &\left(1-\frac{\alpha(1-c)}{2}\right)^k \left(\|\bm{y}^{1} - \bar{\bm{y}}^{0}\|^2 - \frac{4cL_y^2}{\alpha(1-c)}(\frac{2}{(1-c)\alpha}-1)\right) \\
       &+ \frac{4cL_y^2}{\alpha(1-c)}(\frac{2}{(1-c)\alpha}-1) \\
       \leq &\left(1-\frac{\alpha(1-c)}{2}\right)\|\bm{y}^{1} - \bar{\bm{y}}^{0}\|^2 \\
       &+ \frac{\alpha(1-c)}{2}*\frac{4cL_y^2}{\alpha(1-c)}(\frac{2}{(1-c)\alpha}-1) \\
       = &\left(1-\frac{\alpha(1-c)}{2}\right)\|\bm{y}^{1} - \bar{\bm{y}}^{0}\|^2 + 2cL_y^2(\frac{2}{(1-c)\alpha}-1)
    \end{aligned}
\end{equation*}
This gives the upper bound on $\{\bm{y}^k\}$:
\begin{equation*}
    \begin{aligned}
        \|\bm{y}^k\| \leq &\frac{\sqrt{N}\|\bm{b}\|_{\infty}}{1-c} + \sqrt{\left(1-\frac{\alpha(1-c)}{2}\right)\|\bm{y}^{1} - \bar{\bm{y}}^{0}\|^2 + 2cL_y^2(\frac{2}{(1-c)\alpha}-1)} \\
        &=: B.
    \end{aligned}
\end{equation*}
Together with Lemma \ref{lemma: property-hatphi}, we get to know that the $L$-Lipschitz continuity of $\widehat{\grd}\Phi(\GSO,\cdot)$ is applicable on the set $\{\bm{y}^k\}$ with $L = \frac{\sqrt{N}\lambda}{1-c} + \frac{\lambda \sqrt{N}Bc}{(1-c)^2}$. This will be used in the proof of the second part of the lemma.

For the second part, from the $L_{\ell}$-smoothness of $l(\cdot)$ (Lemma \ref{lemma: property-l}), we have
\begin{equation}\label{inequ_ingredient_lip_gradient_l}
\begin{aligned}
        l(\GSO^{k+1})-l(\GSO^{k})\leq &\langle\nabla l(\GSO^{k}),\GSO^{k+1}-\GSO^{k}\rangle \\
    &+(L_{\ell}/2)\|\GSO^{k+1}-\GSO^{k}\|^{2}.
\end{aligned}
\end{equation}
Note that
\begin{equation*}
\begin{aligned}
       & \langle\nabla l(\GSO^{k}),\GSO^{k+1}-\GSO^{k}\rangle \\
       = &\langle\nabla l(\GSO^{k})-\widehat{\grd}\Phi\left(\GSO^{k},\bm{y}^{k+1}\right), \GSO^{k+1}-\GSO^{k}\rangle \\
       &+\langle\widehat{\grd}\Phi(\GSO^k,\bm{y}^{k+1}),\GSO^{k+1}-\GSO^{k}\rangle \\
       \leq &\langle\nabla l(\GSO^{k})-\widehat{\grd}\Phi\left(\GSO^{k},\bm{y}^{k+1}\right), \GSO^{k+1}-\GSO^{k}\rangle \\
       &- \frac{1}{\gamma}\|\GSO^{k+1} - \GSO^k\|^2 \\
        \leq &\frac{\gamma}{2}\|\nabla l(\GSO^{k})-\widehat{\grd}\Phi\left(\GSO^{k},\bm{y}^{k+1}\right)\|^2 - \frac{1}{2\gamma}\|\GSO^{k+1} - \GSO^k\|^2 \\
        \leq &\frac{\gamma}{2}L^2\|\bar{\bm{y}}^k - \bm{y}^{k+1}\|^2 - \frac{1}{2\gamma}\|\GSO^{k+1} - \GSO^k\|^2.
\end{aligned}
\end{equation*}
Here, the first inequality is from the update of $\GSO$ and the last inequality is from the Lipschitz continuity of $\widehat{\grd}\Phi\left(\GSO^{k},\cdot\right)$ (Lemma \ref{lemma: property-hatphi}).

Combining the above result with \eqref{inequ_ingredient_lip_gradient_l} gives
\begin{equation*}
    \begin{aligned}
        l(\GSO^{k+1})-l(\GSO^{k}) \leq &-(\frac{1}{2\gamma}-\frac{L_{\ell}}{2})\|\GSO^{k+1}-\GSO^{k}\|^{2} \\
        &+ \frac{\gamma}{2}L^{2}\|\bm{y}^{k+1}-\bar{\bm{y}}^k\|^{2}.
    \end{aligned}
\end{equation*}
\end{proof}

\begin{Prop}\label{prop:inequa-deltax-deltay}
    If we set $\alpha = \frac{\mu_g}{L_g^2} = \frac{1-c}{(1+c)^2}$ and $\gamma \leq \min\{\frac{3}{4L_{\ell}}, \frac{\mu_g}{4LL_y}\alpha\}$, we have
\begin{equation}\label{upperbond_deltax}
    \frac{1}{K}\sum_{k = 1}^K\|\GSO^k - \GSO^{k-1}\|^2 \leq \mathcal{O}(\frac{1}{K})
\end{equation}
\begin{equation}\label{upperbond_deltay}
    \frac{1}{K}\sum_{k = 1}^K\|\bm{y}^k - \bar{\bm{y}}^{k-1}\|^2 \leq \mathcal{O}(\frac{1}{K})
\end{equation}
\end{Prop}
\begin{proof}
    Denote \\
    $\Theta^k = \|\GSO^k - \GSO^{k-1}\|^2$, $\Gamma^k = \|\bm{y}^k - y^*(\GSO^{k-1})\|^2$, $\Omega^k = l(\GSO^k)$, $c_0 = \frac{1}{2\gamma} - \frac{L_{\ell}}{2}$, $c_1 = \frac{\gamma}{2}L^2$, $d_0 = \frac{\alpha\mu_g}{2}$, $d_1 = L_y^2(\frac{2}{\mu_g\alpha} - 1)$. \\
    Since $\gamma < \frac{3}{4L_{\ell}}$,
    \begin{equation*}
        \begin{aligned}
            \frac{c_0}{c_1} &= \frac{1 - \gamma L_{\ell}}{\gamma^2 L^2} \geq \frac{1}{4\gamma^2L^2} \geq \frac{4L_y^2}{\mu_g^2\alpha^2} > L_y^2\frac{4 - 2\mu_g\alpha}{\alpha^2\mu_g^2} = \frac{d_1}{d_0}.
        \end{aligned}
    \end{equation*}
    Applying Proposition \ref{prop:inequality-general} and Proposition \ref{prop:inequality-recursion} gives the results.
\end{proof}

\begin{Prop}\label{prop:upperbund-innerproduct}
        With the stepsizes in Proposition \ref{prop:inequa-deltax-deltay}, we have
    \begin{equation}\label{uperbound_innerproduct}
        \frac{1}{K}\sum_{k = 0}^K \langle \GSO^k - \GSO^{k+1}, \widehat{\grd}\Phi(\GSO^k, \bm{y}^{k+1})\rangle \leq \mathcal{O}(\frac{1}{K})
    \end{equation}
\end{Prop}

\begin{proof}
\begin{equation*}
    \begin{aligned}
        &l(\GSO^{k+1}) \\
        \leq & l(\GSO^{k}) + \langle\nabla l(\GSO^{k}),\GSO^{k+1}-\GSO^{k}\rangle +(L_{\ell}/2)\|\GSO^{k+1}-\GSO^{k}\|^{2} \\
        = & l(\GSO^{k}) +\langle\nabla l(\GSO^{k})-\widehat{\grd}\Phi\left(\GSO^{k},\bm{y}^{k+1}\right), \GSO^{k+1}-\GSO^{k}\rangle \\
        & +\langle\widehat{\grd}\Phi(\GSO,\bm{y}^{k+1}),\GSO^{k+1}-\GSO^{k}\rangle +(L_{\ell}/2)\|\GSO^{k+1}-\GSO^{k}\|^{2} \\
         \leq & l(\GSO^{k}) + \frac{1}{2}\|\nabla l(\GSO^{k})-\widehat{\grd}\Phi\left(\GSO^{k},\bm{y}^{k+1}\right)\|^2 \\
        &+ \frac{1}{2}\|\GSO^{k+1}-\GSO^{k}\|^2 + \langle\widehat{\grd}\Phi(\GSO^k,\bm{y}^{k+1}),\GSO^{k+1}-\GSO^{k}\rangle \\
        &+(L_{\ell}/2)\|\GSO^{k+1}-\GSO^{k}\|^{2} \\
         \leq & l(\GSO^{k}) + \frac{L^2}{2}\|\bm{y}^{k+1}-\bar{\bm{y}}^k\|^2 + \frac{1+L_{\ell}}{2}\|\GSO^{k+1}-\GSO^{k}\|^2 \\
        &+ \langle\widehat{\grd}\Phi(\GSO^k,\bm{y}^{k+1}),\GSO^{k+1}-\GSO^{k}\rangle 
            \end{aligned}
\end{equation*}
\begin{equation*}
    \begin{aligned}
        \Rightarrow &\langle\widehat{\grd}\Phi(\GSO^k,\bm{y}^{k+1}),\GSO^{k} - \GSO^{k+1}\rangle \leq l(\GSO^{k}) - l(\GSO^{k+1})\\
        &+ \frac{L^2}{2}\|\bm{y}^{k+1}-\bar{\bm{y}}^k\|^2 + \frac{1+L_{\ell}}{2}\|\GSO^{k+1}-\GSO^{k}\|^2.
    \end{aligned}
\end{equation*}
Hence, 
\begin{equation*}
\begin{aligned}
        & \frac{1}{K}\sum_{k = 0}^K \langle \GSO^k - \GSO^{k+1}, \widehat{\grd}\Phi\left(\GSO^{k},\bm{y}^{k+1}\right)\rangle \\
    \leq & \frac{(l(\GSO^0) - l(\GSO^*))}{K} + \frac{L^2}{2}\frac{1}{K}\sum_{k = 0}^K\|\bm{y}^{k+1}-\bar{\bm{y}}^k\|^2\\
    & + \frac{1+L_{\ell}}{2}\frac{1}{K}\sum_{k = 0}^K\|\GSO^{k+1}-\GSO^{k}\|^2 \\
    \leq & \mathcal{O}(\frac{1}{K})
\end{aligned}
\end{equation*}
Here, we apply Proposition \ref{prop:inequa-deltax-deltay} in the last inequality.
\end{proof}

Before going deep in the final part of the proof of Theorem \ref{thm:ttgd}, let us introduce some new notations from optimization theory that would be useful in the non-convex analysis.
\begin{Def}\label{def:moreau-env}
    Fix $\rho > 0$, define the Moreau envelop and proximal map as
    \begin{equation*}
        \begin{aligned}
            &\phi_{1/\rho}(\bm{Z}):=\min_{\GSO\in \mathcal{S}}\{\ell(\GSO)+(\rho/2)\|\GSO - \bm{Z}\|^2\} \\
            &\hat{x}(\bm{Z}) := \argmin_{\GSO\in \mathcal{S}}\{\ell(\GSO)+(\rho/2)\|\GSO - \bm{Z}\|^2\}
        \end{aligned}
    \end{equation*}
\end{Def}
The proof of Lemma 3.8 in \cite{hong2023two} gives the following properties of Moreau envelop and proximal map  in our TTGA update.
\begin{Lemma}
    Set $\rho > \max\{0, -L_{\ell}\}$, we have
        \begin{equation}\label{property-mv1}
        \begin{aligned}
            \phi_{1/\rho}(\GSO^{k+1}) \leq &\phi_{1/\rho}(\GSO^k) + \frac{3}{2}\rho\|\GSO^{k+1} - \GSO^k\|^2 \\
            &+ \rho\gamma \langle \hat{x}(\GSO^k) - \GSO^k, \widehat{\grd}\Phi(\GSO^k, \bm{y}^{k+1})\rangle \\
            &+ \rho\gamma\langle \GSO^k - \GSO^{k+1}, \widehat{\grd}\Phi(\GSO^k, \bm{y}^{k+1})\rangle
        \end{aligned}
        \end{equation}
        and
        \begin{equation}\label{property-mv2}
        \begin{aligned}
            -\frac{\mu_l+\rho}{2}\|\hat{x}(\GSO^k) - \GSO^k\|^2 \geq \langle\nabla l(\GSO^k), \hat{x}(\GSO^k) - \GSO^k\rangle
        \end{aligned}
    \end{equation}
\end{Lemma}
\begin{Theorem}\label{thm:proximal-map-convergency rate}
   With the fixed stepsizes in Proposition \ref{prop:inequa-deltax-deltay} and $\rho > \max\{0,-L_{\ell}\}$, we have
\begin{equation*}
\frac{1}{K}\sum_{k = 0}^K \|\hat{x}(\GSO^k) - \GSO^k\|^2 \leq \mathcal{O}(\frac{1}{K})
\end{equation*}
\end{Theorem}
\begin{proof}
    Note that
    \begin{equation*}
        \begin{aligned}
            & \langle \hat{x}(\GSO^k) - \GSO^k, \widehat{\grd}\Phi(\GSO^k, \bm{y}^{k+1})\rangle = \langle \hat{x}(\GSO^k) - \GSO^k, \nabla l(\GSO^k)\rangle \\
            &\ \ + \langle \hat{x}(\GSO^k) - \GSO^k, \widehat{\grd}\Phi(\GSO^k, \bm{y}^{k+1}) - \nabla l(\GSO^k)\rangle \\
            \leq &
             \langle \hat{x}(\GSO^k) - \GSO^k, \nabla l(\GSO^k)\rangle + \frac{1}{\rho + \mu_l}\|\widehat{\grd}\Phi(\GSO^k, \bm{y}^{k+1}) - \nabla l(\GSO^k)\|^2 \\
             & \ \ + \frac{\rho + \mu_l}{4}\|\hat{x}(\GSO^k) - \GSO^k\|^2\\
             \leq &
             \frac{L^2}{\rho + \mu_l}\|\bm{y}^{k+1} - \bar{\bm{y}}^k\|^2 - \frac{\rho + \mu_l}{4}\|\hat{x}(\GSO^k) - \GSO^k\|^2.
        \end{aligned}
    \end{equation*}
    Here, the last inequality is due to \eqref{property-mv2}. \\
    Combining the above result with \eqref{property-mv1} gives
    \begin{equation*}
        \begin{aligned}
            \phi_{1/\rho}(\GSO^{k+1}) \leq &\phi_{1/\rho}(\GSO^k) + \frac{3}{2}\rho\|\GSO^{k+1} - \GSO^k\|^2\\
            &+ \frac{\rho\gamma L^2}{\rho + \mu_l}\|\bm{y}^{k+1} - \bar{\bm{y}}^k\|^2  \\
            &+ \gamma\rho\left\langle \GSO^k - \GSO^{k+1}, \widehat{\grd}\Phi\left(\GSO^{k},\bm{y}^{k+1}\right)\right\rangle \\
            &- \frac{\rho\gamma(\rho + \mu_l)}{4}\|\hat{x}(\GSO^k) - \GSO^k\|^2 
        \end{aligned}
    \end{equation*}
        \begin{equation*}
        \begin{aligned}
             \Rightarrow & \frac{1}{K}\sum_{k = 0}^K\|\hat{x}(\GSO^k) - \GSO^k\|^2 \leq  \frac{\phi_{1/\rho}(\GSO^{0})-\Phi_{1/\rho}(\GSO^{K+1})}{K} \\
            & + \frac{3\rho}{2} \frac{1}{K}\sum_{k = 0}^K\|\GSO^{k+1}-\GSO^k\|^2 + \frac{\rho\gamma L^2}{\rho + \mu_l}\frac{1}{K}\sum_{k = 0}^K\|\bm{y}^{k+1} - \bar{\bm{y}}^k\|^2 \\
            & + \gamma\rho\frac{1}{K}\sum_{k = 0}^K\langle \GSO^k - \GSO^{k+1}, \widehat{\grd}\Phi\left(\GSO^{k},\bm{y}^{k+1}\right)\rangle \\
            & \leq \mathcal{O}(\frac{1}{K}).
        \end{aligned}
    \end{equation*}
The last inequality is due to that each components are in the order of $\mathcal{O}(\frac{1}{K})$ (Proposition \ref{prop:inequa-deltax-deltay} \& \ref{prop:upperbund-innerproduct}).
\end{proof}

Apart from the error bound characterized by proximal mapping, we can derive the convergence rate for the error bound given by the neighbouring updates from gradient descent, which is formally stated in the following theorem.
\begin{Theorem}
With the fixed stepsizes in Proposition \ref{prop:inequa-deltax-deltay},
    \begin{equation}
        \begin{aligned}
\min_{k=1,...,K} \| \gamma^{-1} ( \GSO^k - {\sf Proj}_{\cal S} ( \GSO^k - \gamma \grd \ell( \GSO^k ) ) ) \|^2 \leq {\cal O}(\frac{1}{K})
        \end{aligned}
    \end{equation}
\end{Theorem}
\begin{proof}
    \begin{equation*}
        \begin{aligned}
    &\|( \GSO^k - {\sf Proj}_{\cal S} ( \GSO^k - \gamma \grd \ell( \GSO^k ) ) ) \| \\
    \leq &\|( \GSO^k -  {\sf Proj}_{ \cal S } ( \GSO^k - \gamma \widehat{\grd} \Phi( \GSO^k , {\bm y}^{k+1} )))\|\\
    &+ \|({\sf Proj}_{ \cal S } ( \GSO^k - \gamma \widehat{\grd} \Phi( \GSO^k , {\bm y}^{k+1} )) - {\sf Proj}_{\cal S} ( \GSO^k - \gamma \grd \ell( \GSO^k ) ))\| \\
    \leq & \|( \GSO^k -  \GSO^{k+1})\| + \gamma\| \widehat{\grd} \Phi( \GSO^k , {\bm y}^{k+1} ) - \grd \ell( \GSO^k )\| \\
    \leq & \|( \GSO^k -  \GSO^{k+1})\| + \gamma L\|{\bm y}^{k+1} - \bar{\bm y}^{k}\|
        \end{aligned}
    \end{equation*}
    Hence, 
    \begin{equation*}
        \begin{aligned}
             &\| \gamma^{-1} ( \GSO^k - {\sf Proj}_{\cal S} ( \GSO^k - \gamma \grd \ell( \GSO^k ) ) ) \|^2 \\
             \leq &\frac{2}{\gamma^2}\|( \GSO^k -  \GSO^{k+1})\|^2 + 2L^2\|{\bm y}^{k+1} - \bar{\bm y}^{k}\|^2.
        \end{aligned}
    \end{equation*}
    Applying Proposition \ref{prop:inequa-deltax-deltay} implies that
    \begin{equation*}
        \begin{aligned}
            &\frac{1}{K}\sum_{k = 1}^K \| \gamma^{-1} ( \GSO^k - {\sf Proj}_{\cal S} ( \GSO^k - \gamma \grd \ell( \GSO^k ) ) ) \|^2 \leq \mathcal{O}(\frac{1}{K}) \\
            \Rightarrow &\min_{k=1,...,K} \| \gamma^{-1} ( \GSO^k - {\sf Proj}_{\cal S} ( \GSO^k - \gamma \grd \ell( \GSO^k ) ) ) \|^2 \leq {\cal O}(\frac{1}{K})
        \end{aligned}
    \end{equation*}
\end{proof}

\newpage




\end{document}